\documentclass[twocolumn,notitlepage,aps,prb,superscriptaddress]{revtex4-1}

\usepackage{amsthm}
\usepackage{amsmath,amscd,latexsym,amssymb,verbatim,enumerate,graphicx}
\usepackage{hyperref}
\usepackage{color}
\usepackage[colorinlistoftodos]{todonotes}
\usepackage{MnSymbol}
\usepackage[all]{xy}

\newtheorem{thm}{Theorem}
\newtheorem*{thm*}{Theorem}
\newtheorem{definition}{Definition}
\newtheorem{prop}{Proposition}
\newtheorem*{prop*}{Proposition}
\newtheorem{lem}{Lemma}
\newtheorem*{lem*}{Lemma}
\newtheorem{cor}{Corollary}
\newtheorem*{cor*}{Corollary}

\newtheorem{remark}{Remark}

\definecolor{KKgreen}{RGB}{0,100,0}

\newcommand{\ot}{\otimes}

\def\supp{{\rm supp}}

\def\cA{{\mathcal A}}
\def\cB{{\mathcal B}}

\def\cE{{\mathcal E}}
\def\cF{{\mathcal F}}

\def\cH{{\mathcal H}}

\def\cM{{\mathcal M}}
\def\cN{{\mathcal N}}

\def\cR{{\mathcal R}}
\def\cS{{\mathcal S}}

\def\cZ{{\mathcal Z}}

\def\tr{{\rm tr}}
\def\id{{\rm id}}

\def\C{\mathbb{C}}
\def\>{\rangle}
\def\<{\langle}

\def\supp{{\rm supp}}
\def\spann{{\rm span}}
\def\alg{{\rm Alg}}
\def\dim{{\rm dim}}

\newcommand{\ket}[1]{\left|#1\right\rangle}
\newcommand{\bra}[1]{\left\langle#1\right|}
\newcommand{\ketr}[1]{\left|#1\right)}
\newcommand{\brar}[1]{\left(#1\right|}
\newcommand{\bk}[1]{\left(#1\right)}
\newcommand{\lr}[1]{\left #1\right}

\begin{document}

\title{Exact renormalization group flow for matrix product density operators}

\author{Kohtaro Kato}
\affiliation{Department of Mathematical Informatics, \\
Graduate School of Informatics, \\
Nagoya University, Nagoya 464-0814, Japan}

\begin{abstract}
Matrix product density operator (MPDO) provides an efficient tensor network representation of mixed states on one-dimensional quantum many-body systems. We study a real-space renormalization group transformation of MPDOs represented by a circuit of local quantum channels. We require that the renormalization group flow is exact, in the sense that it exactly preserves the correlation between the coarse-grained sites and is therefore invertible by another circuit of local quantum channels. Unlike matrix product states (MPS), which always have a well-defined isometric renormalization transformation, we show that general MPDOs do not necessarily admit a converging exact renormalization group flow. We then introduce a subclass of MPDOs with a well-defined renormalization group flow, and show the structure of the MPDOs in the subclass as a representation of a pre-bialgebra structure. As a result, such MPDOs obey generalized symmetry represented by matrix product operator algebras associated with the pre-bialgebra. We also discuss implications with the classification of mixed-state quantum phases. 
\end{abstract}
\maketitle
%========================================================%
\section{Introduction}
Tensor networks provide a powerful and versatile framework for studying quantum many-body systems, offering efficient descriptions that capture essential features of entanglement and correlations. For one-dimensional (1D) systems, matrix product states (MPSs)~\cite{DMRG1,DMRG2,affleck2004rigorous,mps07} serve as an efficient and compact tensor network representation of the pure ground states of gapped local Hamiltonians~\cite{Hastings_2007}. Moreover, any MPS has a corresponding gapped, local, and frustration-free parent Hamiltonian such that the given MPS is a ground state of this Hamiltonian~\cite{mps07,fannes1992finitely}. This equivalence provides a physical interpretation of the descriptive capability of MPS.

One of the fundamental problems in many-body physics is classifying various phases of matter. Tensor networks can represent ground states of several intriguing quantum phases, such as symmetry-protected topological (SPT) phases~\cite{affleck2004rigorous,PhysRevB.81.064439} and topologically ordered phases~\cite{Pepslevinwen1,Pepslevinwen2,SCHUCH20102153,MPOinjectivePEPS}.  
For gapped systems, two systems are in the same gapped quantum phase if their ground states can be connected by a short-depth local unitary circuit ~\cite{LU1,LU2}. 

Real-space renormalization group (RG) flow is a robust approach for understanding large-scale physics, as it allows the exclusion of short-range details while capturing the universal properties of a system. MPSs have been demonstrated to exhibit a well-defined isometric renormalization flow~\cite{PhysRevLett.94.140601}, which can be interpreted as a short-depth local unitary circuit. As this circuit is exactly invertible, MPSs lie within the same gapped quantum phase as the fixed point. Consequently, it is sufficient to classify the phases of fixed points that exhibit simpler structures~\cite{PhysRevB.83.035107,MPSphase2}. 
%However, for mixed states described by MPDOs, the renormalization process is more intricate. General MPDOs do not always admit a converging renormalization flow, complicating the classification of mixed-state quantum phases.

Hastings initiated the study of mixed-state quantum phases, defined as equivalence classes of physically realizable mixed states connected by a short-depth local circuit of quantum channels~\cite{PhysRevLett.107.210501}. These phases represent a natural generalization of symmetry-protected topological (SPT) phases and topologically ordered phases, and their classification has garnered increasing attention in recent years~\cite{mixedphase0,mixedphase3,mixedphase1,mixedphase2,MPOalg2}.

Matrix product density operators (MPDOs), a generalization of MPSs to one-dimensional (1D) mixed states, are natural candidates for representing physically realizable mixed states. MPDOs efficiently describe a broad class of quantum states, including the Gibbs states of local Hamiltonians~\cite{MPOGibbs1,MPOGibbs2} and the fixed points of local dissipative dynamics~\cite{MPOdissipative}. However, in contrast to MPSs, the full descriptive power of MPDOs remains unclear. While generic MPDOs are expected to represent Gibbs states~\cite{MPDOCMI}, they also encompass boundary states of two-dimensional (2D) topologically ordered phases~\cite{PEPSboundary0,PEPSboundary,CIRAC2017100,MPOalg2}, which cannot be described as Gibbs states of local Hamiltonians.

To explore the structure of MPDOs, Cirac et al.~\cite{CIRAC2017100} introduced the concept of ``renormalization fixed point MPDOs," where the system size can be freely increased or decreased by applying local quantum channels. It has been shown that the structure of these fixed points consists of a product state combined with a global matrix product operator (MPO) that commutes with the product state. The global MPO is believed to capture the global anomalous MPO symmetry of the boundary theory, with its algebraic structure further examined in studies such as Refs.~\cite{MPOinjectivePEPS, MPOalg2}.

One can imagine that general MPDOs are converted to these fixed points through short-depth local circuits of quantum channels and vice versa, and the mixed-state quantum phases are classified by analyzing the fixed points as the gapped quantum phases of MPS. However, the existence of such RG flows is not as straightforward as that of MPS, as quantum channels are generally not invertible.

In this paper, we address the challenge of renormalizing MPDOs by introducing a real-space RG transformation, specifically tailored for these operators. Our approach utilizes a circuit of local quantum channels that precisely preserves correlations between coarse-grained sites, ensuring that the flow is exactly invertible by another set of local quantum channels. We refer to this RG flow as exact~\cite{mixedphase1}. The invertibility condition ensures that the coarse-graining operation performs an exact compression of a quantum mixed bipartite state~\cite{KIdecomposition, Hayden2004-os,kato2023}, a concept extensively studied in the field of sufficient statistics~\cite{Fisher1,Cover2006,petz1986sufficient,10.1093/qmath/39.1.97}.

Our first result shows that, unlike MPSs, MPDOs generally do not admit a converging exact RG flow. This finding motivates us to introduce a subclass of MPDOs that exhibit a well-defined renormalization flow. We further demonstrate that this subclass of MPDOs can be expressed as the product of a global MPO, forming an MPO representation of a pre-bialgebra structure and a locally interacting component that commutes with the global MPO. This structure theorem implies that MPDOs in this subclass obey generalized symmetries described by pre-bialgebras~\cite{MPOalg1}. This discovery enhances our understanding of the symmetry properties of mixed states and their behavior under renormalization.

Finally, we analyze the fixed points of the renormalization flow and discuss their implications for the classification of mixed-state quantum phases, providing new insights into the descriptive capabilities of MPDOs and the classification of mixed-state quantum phases.

The organization of this paper is as follows: In Sec.~\ref{sec:prelim}, we introduce various notations and set up the necessary framework. In Sec.~\ref{sec:convergence}, we analyze the conditions under which exact RG flows converge, showing that certain MPDOs do not admit a converging flow. In Sec~\ref{sec:structure},we introduce a subclass of MPDOs that possess a well-defined exact RG flow and derive the structure theorem of these MPDOs in terms of an MPO representation of pre-bialgebras. Sec.~\ref{sec:exampleX} presents an explicit example of an MPDO that does not lie at a fixed point but exhibits an exact RG flow. As an application of the structure theorem, in Sec.~\ref{sec:symmetry}, we show that MPDOs in the subclass obey a MPO-symmetry. We discuss an implication to classification of 1D mixed-state quantum phases in Sec.~\ref{sec:classification}. we discuss the implications of these findings for the classification of one-dimensional (1D) mixed-state quantum phases. All technical proofs are provided in the Appendix.

\section{Preliminary}\label{sec:prelim}
{\bf Notations:}
For a $C^*$-algebra $\cA\subset\cB(\cH)$, we denote the commutant by $\cA'$ and the dual space by $\cA^*$. For a set of operators $\cS\subset\cB(\cH)$, ${\rm span}\cS$ is the linear span and ${\alg}(\cS)$ is the $C^*-$algebra generated by $\cS$. For an operator $O\in\cB(\cH)$, $\supp(O)$ denotes the support of $O$. 

For a finite-dimensional Hilbert space $\cH$, we denote the set of bounded operators by $\cB(\cH)$ and the set of density operators (positive semi-definite operators with normalized trace) by $\cS(\cH)$. For composite systems, we often distinguish the subsystems by their indices and simply denote $\cH_A\ot\cH_B$ by $\cH_{AB}$. Similarly, the tensor products of vectors are omitted, and we denote by $\ket{i_1}\ot\ket{i_2}$ by $\ket{i_1i_2}$. For a density operator $\rho_{AB}\in\cS(\cH_{AB})$, the reduced density operator on $A$ is written as $\rho_A:=\tr_B\rho_{AB}$.   
We refer to completely-positive trace-preserving (CPTP) maps as quantum channels. The von Neumann entropy of a state $\rho$ is defined by $S(\rho):=\tr\bk{\rho\log_2\rho}$. The mutual information of a bipartite state $\rho_{AB}$ is then defined as $I(A:B)_\rho:=S(\rho_A)+S(\rho_B)-S(\rho_{AB})\geq0$. 

We denote a finite 1D spin system with length $L$ by $\Lambda_L\subset \mathbb{Z}$, where a quantum spin $\C^d (d<\infty)$ is associated for each site. The computational basis of $\C^d$ is denoted by $\{\ket{i}\}_{i=1}^d$. Hilbert space $\C^D$ is used to represent the virtual degrees of freedom of tensor networks. To distinguish it from physical degrees of freedom $\C^d$, the computational basis of $\C^D$ is denoted by $\{\ketr{\alpha}\}_{\alpha=1}^D$. 
 
    \subsection{Matrix Product States and Matrix Product Density Operators}
        \subsubsection{Matrix Product States (MPSs)}
        Consider a tensor $A$ defined on $\C^d\ot\cB\bk{\C^D}$,
        \begin{equation*}
            A=\sum_{i=1}^d\ket{i}\ot A^i\,,
        \end{equation*}
        where $A^i\:(i=1,...,d)$ are $D\times D$ matrices.  The dimension $D$ is called the bond dimension of $A$ and depends on the choice of the tensor. Note that throughout the paper, we consider only constant bond dimension $D$. 
        We define pure quantum states as:  
        \begin{equation*}
            \ket{\psi^{(L)}}=\sum_{i_1,...,i_L=1}^d\tr\bk{A^{i_1}A^{i_2}...A^{i_L}}\ket{i_1i_2...i_L},
        \end{equation*} 
        for $L=1,2,...$ are the (translation-invariant) MPS generated by $A$ (Fig.~\ref{fig:MPS1}).

\begin{figure}[htbp]
	\centering
	\includegraphics[width=0.48\textwidth]{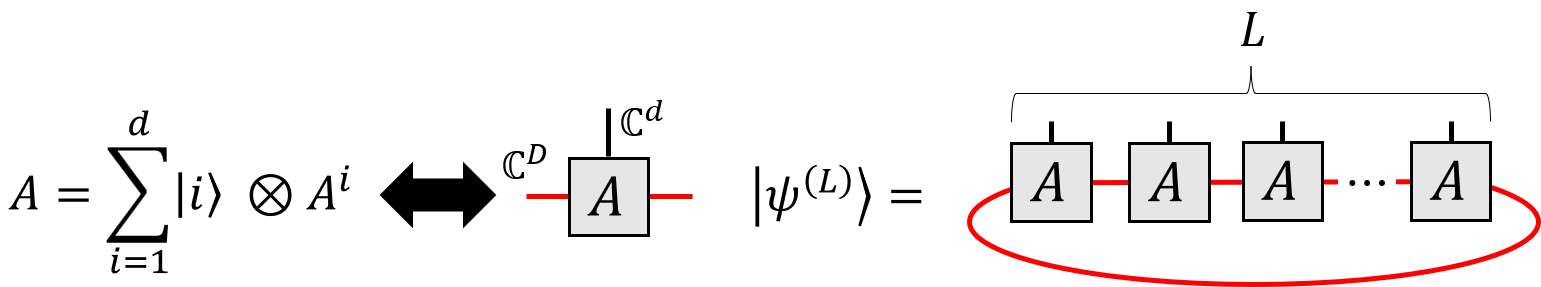}
	\caption{A tensor for MPS is diagrammatically represented by a three-legged box. MPS (with the closed boundary condition) is then represented by concatenating the virtual degrees of freedom (red) in a horizontal direction. }
	\label{fig:MPS1}
\end{figure}

        MPS is invariant under a similarity transformation 
        \begin{equation*}
            A^i\mapsto XA^iX^{-1},
        \end{equation*}
        where $X$ is any invertible matrix. In this way, any tensor $A$ on $\C^d\ot\cB\bk{\C^D}$ can be transformed to the form:
        \begin{equation}\label{eq:MPSCF}
            A^i=\bigoplus_{k=1}^g A_k^i\ot N_k\,,
        \end{equation}
        where $N_k={\rm diag}\bk{\mu_1^k,\mu_2^k,...,\mu_{D_k^R}^k}$ is a $D_k^R$-dimensional diagonal matrix and $\{A_k^i\}_k$ is a set of $D_k^L\times D_k^L$ matrices forming a {\it basis of normal tensors (BNT)} of $A$~\cite{CIRAC2017100}. The dimensions of factors $D_k^L, D_k^R$ and the matrices $N_k$ depend on the choice of $A$. 
        \begin{definition}
        A set of matrices $\{A_k^i\}_k$ form a \it{basis of normal tensors} (BNT) of $A$ if the following condtions hold:
        \begin{itemize}
            \item[(i)] Vectors
            \begin{equation*}
                \sum_{i_1,...,i_L=1}^d\tr\bk{A_k^{i_1}A_k^{i_2}...A_k^{i_L}}\ket{i_1i_2...i_L}\,,\quad k=1,...,g
            \end{equation*}
            are linearly independent for all $L\in\mathbb{N}$.
            \item[(ii)] The MPS generated by $A$ is a linear combination of these vectors.
            \item[(iii)] Each $\cE_k(\cdot):=\sum_iA^i\cdot \bk{A^i}^\dagger$ is a completely-positive (CP) map with the maximal non-degenerate eigenvalue 1 which has no nontrivial invariant subspace.
        \end{itemize}        
        \end{definition}
        
        It is known that any tensor satisfies the following additional condition after blocking a certain number of times, independent of $L$~\cite{mps07}:
        \begin{definition}
            We say $A$ is in \it{block injective canonical form (biCF)} if it is in the form~\eqref{eq:MPSCF} and for any 
            $X=\bigoplus_kX_k$, there exist $c_{i}(X)\in\C$ such that 
            \begin{equation}\label{eq:biCF}
                X_k=\sum_{i}c_{i}(X)A_k^{i}\,\quad k=1,...,g.
            \end{equation}
        \end{definition}
        In what follows, we assume that the MPS tensor~\eqref{eq:MPSCF} is in biCF without loss of generality.  
        
        \subsubsection{Matrix Product Density Operators (MPDOs)}
        An MPDO is a generalization of MPS to 1D mixed states. 
        A tensor $M\in\cB(\C^d\ot\C^D)$ for an MPDO has four legs and written as:
        \begin{equation*}
            M=\sum_{i,j=1}^d\ket{i}\bra{j}\ot M^{ij}\,,
        \end{equation*}
        where $M^{ij}\:(i,j=1,...,d)$ are $D\times D$ matrices. We define the density operators:  
        \begin{equation}\label{eq:MPDO}
            \rho^{(L)}=\sum_{\bf i,j} \tr\bk{M^{i_1j_1}M^{i_2j_2}...M^{i_Lj_L}}\ket{i_1...i_L}\bra{j_1...j_L}\,,
        \end{equation}
        where ${\bf i}=(i_1,...,i_L)$ and ${\bf j}=(j_1,...,j_L)$, for $L=1,2,...$ are the MPDO generated by $M$ (Fig.~\ref{fig:MPDO}). Note that while some literature does not require the normalization $\tr\rho^{(L)}=1$ we enforce it in this paper.  
 \begin{figure}[htbp]
	\centering
	\includegraphics[width=0.48\textwidth]{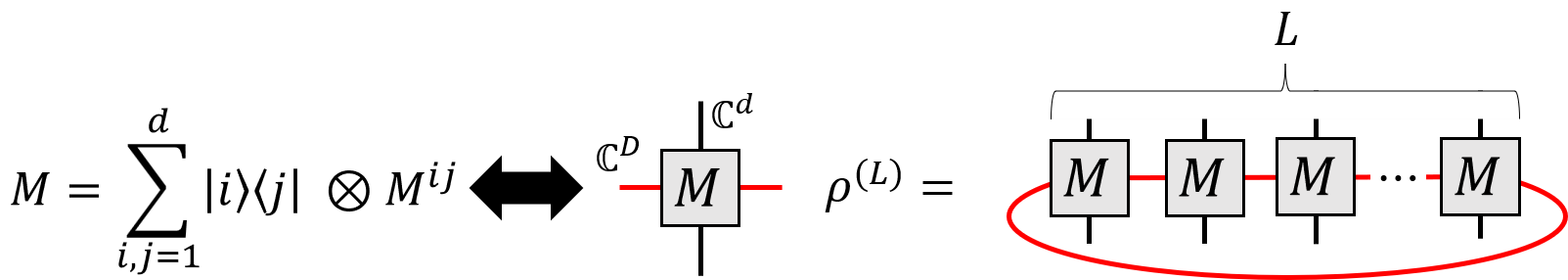}
	\caption{A tensor for MPDO is diagrammatically represented by a four-legged box, where the upper leg represents the ket and the bottom leg represents the bra vector. An MPDO (with the closed boundary condition) is then represented by concatenating the virtual degrees of freedom (red) in a horizontal direction. }
	\label{fig:MPDO}
\end{figure}

        Similarly to the MPS, any tensor $M$ can be transformed into the form:
        \begin{equation}\label{eq:h-CF1}
            M^{ij}=\bigoplus_{k=1}^g M^{ij}_k\ot N_k\,,
        \end{equation}
        where $N_k={\rm diag}\bk{\mu_1^k,\mu_2^k,...,\mu_{D_k^R}^k}$ is a $D_k^R$-dimensional diagonal matrix and $\{M_k^{ij}\}_k$ is a set of $D_k^L\times D_k^L$ matrices forming a BNT. As for MPSs, we assume that $M$ is also biCF without loss of generality. With this additional assumption, we call the form in Eq.~\eqref{eq:h-CF1} a {\it horizontal canonical form (h-CF)} (Fig.~\ref{fig:hCF}). 

\begin{figure}[htbp]
	\centering
	\includegraphics[width=0.35\textwidth]{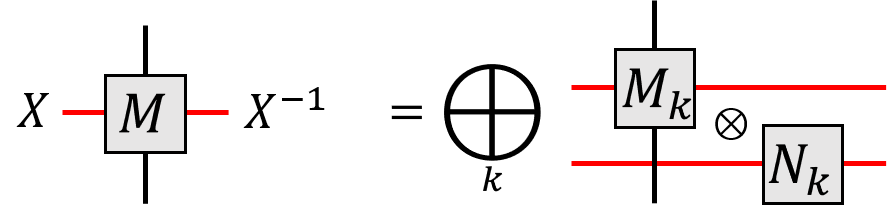}
	\caption{The h-CF of the MPDO tensor $M$. After the decomposition, each $N_k$ only has virtual legs and does not depend on the physical degrees of freedom. }
	\label{fig:hCF}
\end{figure}

        Since $M$ can be multiplied not only in the horizontal direction (virtual degrees of freedom), but also in the vertical direction (physical degrees of freedom), it is convenient to introduce the following representation: \begin{align}
            M&=\sum_{k=1}^g\sum_{i,j=1}^d\ket{i}\bra{j}\ot M^{ij}_k\ot N_k \nonumber\\
            &=\sum_{k=1}^g\sum_{\alpha^k,\beta^k=1}^{D_k^L} W_k^{\alpha_k\beta_k} \ot \ketr{\alpha_k}\brar{\beta_k}\ot N_k\,,
        \end{align}
        where 
        \begin{equation*} W_k^{\alpha_k\beta_k}:=\sum_{i,j=1}^d\brar{\alpha_k}M^{ij}_k\ketr{\beta_k}\ket{i}\bra{j}\,.
        \end{equation*}
        $W^{\alpha_k\beta_k}_k$ represent the matrices that appear when $M$ is regarded as an MPO tensor in the vertical direction (Fig.~\ref{fig:MW}). 
 \begin{figure}[htbp]
	\centering
	\includegraphics[width=0.4\textwidth]{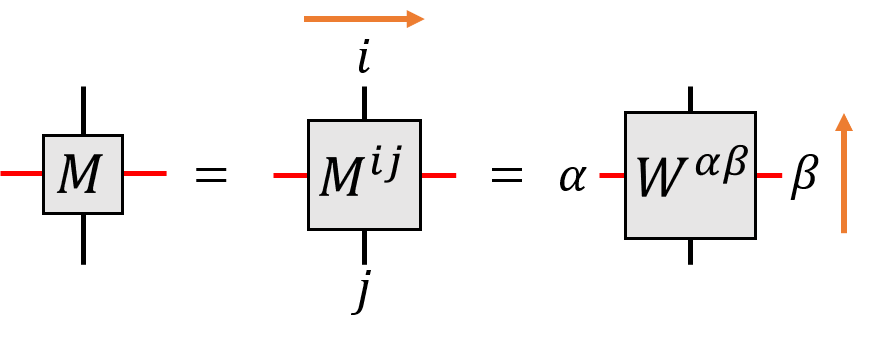}
	\caption{The MPDO tensor $M$ can be regarded as matrices acting on the virtual degrees of freedom ($M$) or as matrices acting on the physical degrees of freedom ($W$). Here we first decompose $M^{ij}$ to the h-CF to eliminate irrelevant degrees $N_k$, and then consider each $M_k^{ij}$ as a matrix in the vertical direction. }
	\label{fig:MW}
\end{figure}
        
        In terms of these matrices $W_k^{\alpha_k\beta_k}$, $\rho^{(L)}$ is written as
        \begin{equation*}
            \rho^{(L)}=\sum_{k=1}^g\tr(N_k^L)\hspace{-0.2cm}\sum_{\alpha_k^1,...,\alpha_k^L=1}^{D_k^L}W_k^{\alpha_k^1\alpha_k^2}\ot W_k^{\alpha_k^2\alpha_k^3}\ot\cdots\ot  W_k^{\alpha_k^L\alpha_k^1}\,.
        \end{equation*}

        It has been shown~\cite{CIRAC2017100} that for any MPDOs in h-CF, there is a unitary $U$ such that the following equation holds: 
        \begin{align*}
            UW_k^{\alpha_k\beta_k}U^\dagger=\bigoplus_{a=1}^r W_{k,a}^{\alpha_k\beta_k} \ot \Omega_a\,,
        \end{align*}
        where $\{W_{k,a}^{\alpha_k\beta_k}\}_a$ forms a BNT and $\Omega_a>0$. We call this decomposition as the {\it vertical canonical form (v-CF)} (Fig.~\ref{fig:vCF}).
\begin{figure}[htbp]
	\centering
	\includegraphics[width=0.3\textwidth]{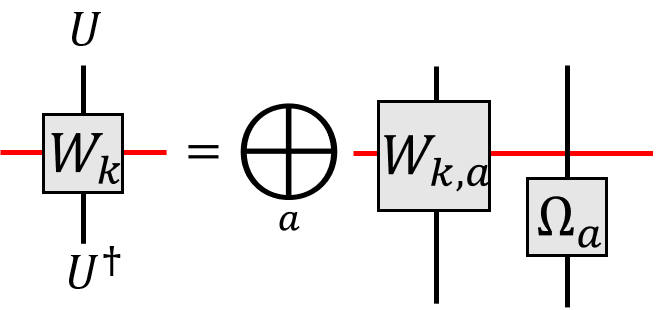}
	\caption{The v-CF of the MPDO tensor $M$ for each $W_k$. The v-CF can be achieved by a unitary transformation, and each $\Omega_a$ only has physical legs and does not genuinely contribute to the correlation with the neighbor tensors. }
	\label{fig:vCF}
\end{figure}
        
        Unlike h-CF, v-CF does not simply concatenate under tensor blockings. Consider the vertical matrices over a block of $l$ sites: 
        \begin{equation*}
            W_k^{\alpha_k\beta_k}[l]:=\sum_{\gamma_k^1,..,\gamma_k^{l-1}=1}^{D_k^L}W_k^{\alpha_k\gamma^1_k}\ot W_k^{\gamma^1_k\gamma_k^2}\ot\cdots\ot W_k^{\gamma^{l-1}_k\beta_k}\,.
        \end{equation*}
        These $d^l\times d^l$ matrices are also transformed by a unitary matrix $U_l$ on $\bk{\C^d}^{\otimes l}$ to the v-CF:
        \begin{align}
            U_lW_k^{\alpha_k\beta_k}[l]U_l^\dagger=\bigoplus_{a'=1}^{r'} W_{k,a'}^{\alpha_k\beta_k}[l] \ot \Omega_{a'}^{(l)}\,.\label{eq:v-CF}
        \end{align}
        In general, both the number of blocks $r'$ and the sizes of the direct sums of blocks depend on $l$. Note that we do not require biCF condition~\eqref{eq:biCF} in the vertical direction, as blocking tensors vertically would change the state. 
        
    \subsection{Exact RG transformations of 1D tensor networks}
    Real-space RG transformation is a coarse-graining process which removes the effects of short-range physics and provides information about large-scale (long-range) physics such as, the phases of matter. 
    
    In this paper, we focus on {\it exact} real-space RG transformations, meaning that each RG step completely preserves correlations between blocks and is therefore invertible by another quantum channel (fine-graining). Although approximate RG transformations are more 
 practical for numerical use, the relationships to well-defined fixed points and phase classifications are more subtle.  

        \subsubsection{RG transformations of Matrix Product States}
        An exact real-space RG transformation for MPS was first introduced in Ref.~\cite{PhysRevLett.94.140601}. For MPSs, the von Neumann entropy of $\rho_A=\tr_{\Lambda_L\backslash A}\psi^{(L)}$, $S(\rho_A)$ is bounded by a constant (an area law of entanglement). The von Neumann entropy is the asymptotically optimal rate of quantum data compression~\cite{PhysRevA.51.2738}, and thus the area law suggests that effective degrees of freedom can be reduced to a fixed size. After blocking neighboring sites such that $d\geq D^2$, the polar decomposition of the tensors yields that there always exists an isometry $V:\C^d\to(\C^d)^{\ot 2}$ and matrices ${\tilde A}^k$ such that: 
        \begin{equation}
            \sum_{i,j=1}^d\ket{ij}\ot A^iA^j=\sum_{k=1}^dV\ket{k}\ot {\tilde A}^k\,
        \end{equation}
        (Fig.~\ref{fig:MPSRG}). 
\begin{figure}[htbp]
	\centering
	\includegraphics[width=0.48\textwidth]{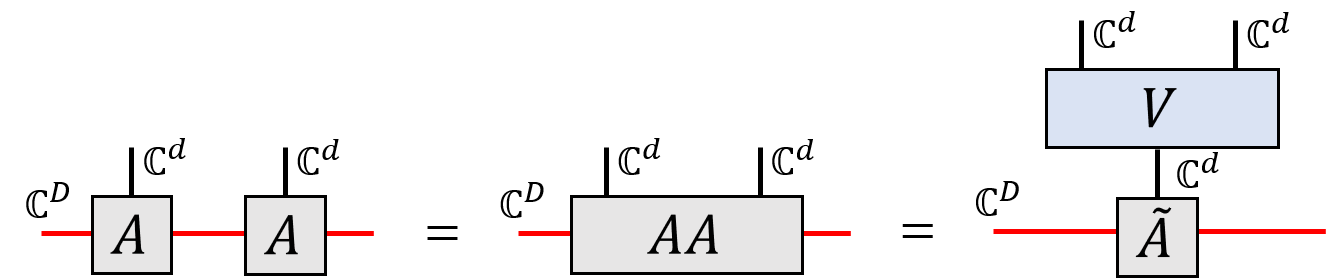}
	\caption{Any MPS admits an isometric RG transformation that reduces the physical system size to that of the single site.}
	\label{fig:MPSRG}
\end{figure}

        One can coarse-grain neighboring $l$-sites into a single site by iterating the same procedure for arbitrary $l$.   
        Therefore, any MPS admits an RG flow that keeps the physical degrees of freedom at each site constant $d$. The transfer matrix $\mathbb{E}:=\sum_{i}A^i\ot\bar{A^i}$ converges to a projector under this coarse-graining~\cite{PhysRevLett.94.140601}.
        
        The fixed points of the RG flow are isometric MPSs~\cite{PhysRevLett.94.140601,MPSphase2}.  Up to a local unitary, an isometric MPS is mapped to pairs of maximally entangled states between sites and a GHZ-like state spreading across the entire system. This implies that 1D gapped quantum phases without symmetry are specified by the ground state degeneracy~\cite{MPSphase2}. 
        
        \subsubsection{RG transformations of MPDOs}\label{sec:RFP1}
        The real-space RG transformation for mixed states is more involved, as typical mixed states, such as the Gibbs states, obey a volume law of entropy.  As a result, a naive blocking and local unitary transformation would lead to divergent degrees of freedom on each site. Therefore, to obtain a converging RG flow, one needs to apply a general quantum channel that properly removes the short-range entropy contributions. 

        In each step of the RG transformation for MPDOs, it is desirable that consecutive spins are mapped to one coarse-grained site by a quantum channel $\cE_{2\to1}$ (Fig.~\ref{fig:MPDORG}). 
        \begin{equation*}
            \sum_{{\bf i},{\bf j}=1}^d\cE_{2\to1}\bk{\ket{i_1i_2}\bra{j_1j_2}}\ot M^{i_1j_1}M^{i_2j_2}=\sum_{i,j}^d\ket{i}\bra{j}\ot {\tilde M}^{ij}\,.
        \end{equation*}
        For an exact RG transformation, there must exist another quantum channel $\cF_{1\to2}$ such that
        \begin{equation*}
            \sum_{i,j}^d\cF_{1\to2}\bk{\ket{i}\bra{j}}\ot {\tilde M}^{ij}=\sum_{{\bf i},{\bf j}=1}^d\ket{i_1i_2}\bra{j_1j_2}\ot M^{i_1j_1}M^{i_2j_2}\,.
        \end{equation*}
\begin{figure}[htbp]
	\centering
	\includegraphics[width=0.38\textwidth]{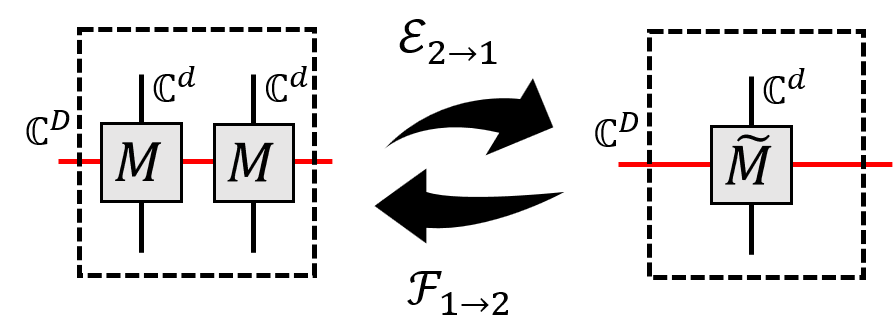}
	\caption{An exact RG transformation for MPDO is a quantum channel $\cE_{2\to 1}$ that reduces the size of the physical system, while keeping the relevant correlation such that there is another quantum channel for which the renormalization process can inverted. }
	\label{fig:MPDORG}
\end{figure}

        However, unlike MPS, the existence of such an RG transformation is not straightforwardly guaranteed for general MPDOs. In fact, in this paper, we show that there is an explicit counterexample of an MPDO that does not admit a converging RG flow (Sec.~\ref{sec:convergence}). 

        An MPDO is a fixed point of a real-space RG transformation if the tensor after  coarse-graining remains the same as before, i.e., there are quantum channels $\cE,\cF$ such that ${\tilde M}^{ij}=M^{ij}$. These fixed points were first introduced in Ref.~\cite{CIRAC2017100}, and it has been shown that, under one additional condition, such fixed point MPDOs can be written as
        \begin{equation}\label{eq:RFPMPDO}
            \rho^{(L)}\cong\sum_{i=1}^d\lambda_iP_i^{(L)}\Omega^{\otimes L}\,,
        \end{equation}
        where $\lambda_i\geq0$, $P_i^{(L)}$ are the MPO projectors on $\Lambda_L$ ($\cong$ meaning the two sides are equal up to a unitary transformation). Although not mathematically specified, $P_i$ is argued to contain the ``topological part" (such as the boundary MPO-symmetry of a topological PEPS), and $\Omega>0$ satisfies $[P_i^{(L)}, \Omega^{\ot L}]=0$~\footnote{In the statement in Theorem 4.15. in Ref.~\cite{CIRAC2017100}, the state is written as $\sum_{i=1}^d\lambda_iP_i^{(L)}e^{-H_L}$ where $H_N$ is a nearest-neighbor commuting Hamiltonian. However, the explicit construction in the proof tells us that $e^{-H_L}$ is given as $\mu^{\otimes L}$ for $\mu>0$ (in this paper we use $\Omega$ instead of $\mu$), which is not interacting~\cite{private}.}.

    \subsection{Exact local compression of mixed bipartite states}
    To study the RG transformation of MPDOs, we introduce some known facts on the exact local compression of general bipartite quantum states by a quantum channel on a subsystem. 
        \begin{definition}~\cite{kato2023}
            We say a quantum channel $\cE_{B\to {\tilde B}}:\cB(\cH_B)\to\cB(\cH_{\tilde B})$ is an exact local compression of $\rho_{AB}$ if there exists another quantum channel $\cR_{{\tilde B}\to B}:\cB(\cH_{\tilde B})\to\cB(\cH_B)$ such that
            \begin{equation*}
                \rho_{AB}=\bk{\cR_{{\tilde B}\to B}\circ\cE_{B\to {\tilde B}}}(\rho_{AB})\,.
            \end{equation*}
            Note that we omit identity maps for simplicity and denote, for example, $\id_A\ot\cE_{B\to {\tilde B}}$ as $\cE_{B\to {\tilde B}}$.
        \end{definition}
        It is well-known~\cite{mixedphase1,kato2023} that $\cE_{B\to {\tilde B}}$ is an exact local compression of $\rho_{AB}$ if and only if
        \begin{equation*}
            I(A:B)_\rho=I\bk{A:{\tilde B}}_{\cE(\rho)}\,,
        \end{equation*}
       where the mutual information measures total amount of bipartite correlations, meaning  that exact compression preserves the correlation between two subsystems. Thus, exact local compression is equivalent to a correlation-preserving map, as defined in Ref.~\cite{mixedphase1}.
%        \subsubsection{Koashi-Imoto decomposition}
        
        The minimal dimension of ${\tilde B}$ for the exact local compression of $\rho_{AB}$ is determined by the size of a certain algebra. 
        \begin{thm*}~\cite{KIdecomposition,Hayden2004-os,Mosonyi2004-wo,Jencova2006-on}
        For any bipartite state $\rho_{AB}\in\cS(\cH_A\ot\cH_B)$, where $\supp(\rho_B)=\cH_B$, there exists a unital $C^*$-algebra $\cM^{\min}_\cS\subset\cB(\cH_B)$ such that
        \begin{itemize}
            \item[(i)] $\cH_B\cong \bigoplus_{a=1}^r\cH_{B_a^L}\ot\cH_{B_a^R}$ and
            \begin{equation*}
                \cM^{\min}_\cS\cong \bigoplus_{a=1}^r\cB\bk{\cH_{B_a^L}}\ot I_{\cH_{B_a^R}}\,.
            \end{equation*}
            \item[(ii)]  There exists a unitary $U_{B\to B^LB^R}$ such that $\rho_{AB}$ is decomposed as 
            \begin{equation}\label{eq:KI-dec}
                U_{B\to B^LB^R}\rho_{AB}U_{B\to B^LB^R}^\dagger=\bigoplus_{a=1}^rp_a\rho_{AB_a^L}\ot \omega_{B_a^R}\,,
            \end{equation}
            where $\{p_a\}$ is a probability distribution, $\rho_{AB_a^L}\in \cS(\cH_{AB_a^L})$ and $\omega_{B_a^R}>0\in\cS\bk{\cH_{B_a^R}}$.
            \item[(iii)] For any CPTP-map $\cE_B$ satisfying $\cE_B(\rho_{AB})=\rho_{AB}$, its Stinespring dilation isometry $V_{B\to BE}:\cB(\cH_B)\to\cB(\cH_{BE})$ is decomposed into 
    \begin{equation*}
        V_{B\to BE}=\bigoplus_{a=1}^rI_{B_a^L} \otimes V_{B_a^R\to B_a^R E}
    \end{equation*}
    which satisfies
    \begin{equation*}
        \tr_E\left(V_{B_a^R\to B_a^R E}\:\omega_{B_a^R}(V_{B_a^R\to B_a^R E})^\dagger\right)=\omega_{B_a^R}\,\quad\forall a.
    \end{equation*}
        \end{itemize}
    \end{thm*}
    The decomposition~\eqref{eq:KI-dec} is sometimes called the Koashi-Imoto decomposition~\cite{KIdecomposition,Hayden2004-os}.  The algebra $\cM_\cS^{\min}$ is the minimal sufficient subalgebra~\cite{petz1986sufficient,10.1093/qmath/39.1.97,Jencova2006-on} of the following set of states $\cS$
    \begin{equation}\label{eq:cS}
        \cS:=\lr{\{\mu_B=\frac{\tr_A\bk{O_A\rho_{AB}}}{\tr\bk{O_A\rho_A}}\:\middle|\:\:0\leq O_A\leq I_A}\}\,,
    \end{equation}
    for which $\cM_\cS^{\min}$ is written as
    \begin{equation}\label{eq:defminsufalg}
        \cM_\cS^{\min}={\rm alg}\lr{\{\mu^{it}_B\rho_B^{-it}\:\middle|\:\:\forall \mu_B\in\cS, \forall t\in\mathbb{R}}\}\,.
    \end{equation}

    The minimal dimension $d_{\tilde B}^{\min}$ of the exact local compression is then given by 
    \begin{equation}\label{eq:mindimB}
        d_{\tilde B}^{\min}=\sum_{a=1}^r\dim(\cH_{B_a^L})\,
    \end{equation}
and the corresponding exact compression is given by
\begin{equation*}
    \cE_{B\to B^L}(\rho_{AB})=\bigoplus_{a=1}^rp_a\rho_{AB_a^L}\,,
\end{equation*}
where $\cH_{B^L}:=\bigoplus_{a}\cH_{B_a^L}$.

\section{Convergence of exact RG flow for MPDOs}\label{sec:convergence}
%        \subsection{v-CF and the Koashi-Imoto decomposition}
    Consider a real-space exact RG transformation of MPDOs that coarse-grains $l$-neighboring sites into one new site. For the MPDO RG flow to converge, the degrees of freedom at each coarse-grained site must be bounded at all iterative steps. This means that, for any $l$, there must exist an exact local compression to a fixed-sized Hilbert space. As the minimal dimension under exact local compression is determined by the dimension of the minimal sufficient subalgebra~\eqref{eq:mindimB}, for any $l$ the corresponding minimal sufficient subalgebra must be isomorphic to an algebra whose dimension is bounded by a constant. 

    For translationally invariant MPDOs, the minimal sufficient subalgebras~\footnote{When we speak about the minimal sufficient subalgebra for a partition $AB~\Lambda_L$, we implicitly assume that $\cH_B$ is restricted to $\supp(\rho_B)$ so that $\rho_B>0$.} are  determined by the length of the region, not by the size of the complementary region.
    \begin{lem}\label{lem:Mldefinition}
    Let $AB$ be a bipartition of $\Lambda_L$ such that $B$ is connected and $|B|=l$. Then, $\cM_\cS^{\min}$ of a MPDO $\rho_{AB}^{(L)}$ on $B$ is the same for all $L\geq l+1$.
    \end{lem}
    We thus denote the minimal sufficient subalgebra on a $l$-length segment $\Lambda_l\subset\Lambda_L$ by $\cM_l^{\min}$. Note that the reduced state $\rho^{(L)}_B$ generally depends not only on $l$ but also on $L-l$, the size of the complementary region. We particularly denote the reduced state of an MPDO on $\Lambda_l\subset\Lambda_L$ for $L=l+1$ by $\rho_l$. 

    An important consequence of this lemma is that the decomposition of v-CF is equivalent to the Koashi-Imoto decomposition~\eqref{eq:KI-dec} of the MPDO.
    \begin{lem}\label{vCF=KI}
 Let $AB$ be a bipartition of $\Lambda_L$ such that $B$ is connected and $|B|=l$. Consider the v-CF of the matrix $W_k^{\alpha_k\beta_k}[l]$:
            \begin{align}\label{eq:vcf1}
           W_k^{\alpha_k\beta_k}[l]\cong\bigoplus_{a'=1}^{r'} W_{k,a'}^{\alpha_k\beta_k}[l] \ot \Omega_{a'}^{(l)}\,
        \end{align}
        with the decomposition of the Hilbert space 
        \begin{equation*}
            \cH_B\cong\bigoplus_{a'=1}^{r'}\cH_{B_{a'}^L}\otimes \cH_{B_{a'}^R}\,.
        \end{equation*}
        Then we have
        \begin{equation}\label{eq:kidec}
                        \cM^{\min}_l\cong \bigoplus_{a'=1}^{r'}\cB\bk{\cH_{B_{a'}^L}}\ot I_{\cH_{B_{a'}^R}}\,.
        \end{equation}
        Conversely, in the basis such that Eq.~\eqref{eq:kidec} holds, $W_k^{\alpha_k\beta_k}[l]$ is in the v-CF~\eqref{eq:kidec}.
    \end{lem}

    As mentioned before, a necessary condition for the existence of a converging RG flow is that the minimal dimension~\eqref{eq:mindimB} is finite for all $l\in\mathbb{N}$, i.e., there is a constant ${\tilde d}\in\mathbb{N}$, injective representations $\{\pi_l\}_{l\in\mathbb{N}}$ and $C^*$-algebras $\{\cA_l\}_{l\in\mathbb{N}}$ s.t. 
    \begin{equation}\label{eq:anecessary}
        \cM_l^{\min}=\pi_l(\cA_l)\,,\quad\cA_l\subset\cB\bk{\C^{\tilde d}}\,,\quad\forall l\in\mathbb{N}\,.
    \end{equation}
    Next, we show two examples, one that satisfies this condition and one that does not.

    \subsection{1D quantum Markov chain}
    We say $\rho_{\Lambda_L}\in\cS\bk{\cH_{\Lambda_L}}$ is a 1D (exact) quantum Markov chain if 
    \begin{equation*}
        I\bk{A:C|B}_\rho=0\,
    \end{equation*}
    for any disjoint tripartition $ABC=\Lambda_L$ such that $B$ shields $A$ from $C$, where $I(A:C|B)_\rho:=I(A:BC)_\rho-I(A:B)_\rho$ is the conditional mutual information.
    
    The quantum Hammerslay-Clifford theorem~\cite{leifer2008quantum,brown2012} states that any positive quantum Markov chain $\rho_{\Lambda_L}>0$ is equivalent to the Gibbs state of a nearest-neighbor commuting Hamiltonian:
    \begin{equation*}
        \sigma=\frac{e^{-H_{\Lambda_L}}}{\tr e^{-H_{\Lambda_L}}}\,, \quad H_{\Lambda_L}=\sum_ih_{i,i+1}, \quad [h_{i,i+1},h_{j,j+1}]=0\,,
    \end{equation*}
    which can be represented by MPDOs with a constant bond dimension.

    A notifiable property of a quantum Markov chain is the existence of a local recovery map: for any segment $C\subset \Lambda_L$, 
    \begin{equation}
        \rho_{\Lambda_L}=R_{\cN(C)\to \cN(C)C}\circ\tr_C\bk{\rho_{\Lambda_L}}\,,
    \end{equation}
    where $\cN(C)$ is the neighbor of $C$. 
    This means that if we divide $\Lambda_L$ into $C\cup\cN(C)$ and the rest, $\tr_C:C\cup\cN(C)\to\cN(C)$ is an exact compression to the subsystem $\cN(C)$, whose size is fixed in 1D. In fact, quantum Markov chains are present in the set of the RG fixed points studied in Ref.~\cite{CIRAC2017100}. 
    
    \subsection{MPDOs with no converging RG transformation}
   In contrast, we show that there are MPDOs that do not satisfy the necessary condition~\eqref{eq:anecessary}. Consider a diagonal matrix
    \begin{equation*}
    T:=\left(
    \begin{array}{ccc}
    \lambda_1&0&0\\
    0&\lambda_2&0\\
    0&0&\lambda_3
    \end{array}
    \right)
\end{equation*}
such that $\tr T=0$, $|\lambda_i|\leq1$, $|\lambda_1|\neq|\lambda_2|\neq|\lambda_3|$ and $\sum_i\lambda_i^2=:\lambda$. Then,
    \begin{equation}\label{eq:counterex}
        \rho^{(L)}=\frac{1}{3^L}\left(I^{\ot L}+T^{\ot L}\right)\geq0
    \end{equation}
    is an MPDO on a chain of qutrits generated by tensor    
    \begin{equation*}
        M^{ij}=\frac{\delta_{ij}}{3}\ketr{0}\brar{0}+\frac{\lambda_i\delta_{ij}}{3}\ketr{1}\brar{1} \quad(i,j=0,1,2)\,.
    \end{equation*}
    Note that $\rho^{(L)}$ is a classical state and does not contain entanglement.  
The h-CF of the tensor is given as
    \begin{align*}
        M^{ij}&\cong(M_1^{ij}\otimes N_1)\oplus(M_2^{ij}\otimes N_2)\,,\\
        M_1^{ij}\otimes N_1&=\left(\frac{\delta_{ij}}{\sqrt{3}}\ketr{0}\brar{0}\right)\ot\left(\frac{1}{\sqrt{3}}\ketr{0}\brar{0}\right)\,,\\
        M_2^{ij}\otimes N_2&=\left(\frac{\lambda_i\delta_{ij}}{\sqrt{\lambda}}\ketr{1}\brar{1}\right)\ot\left(\frac{\sqrt{\lambda}}{3}\ketr{1}\brar{1}\right)\,.
    \end{align*}
    The vertical tensors $W_k^{\alpha_k\beta_k}[l]$, for any length $l$, are then obtained as
    \begin{align*}
        W_1^{00}[l]&=\frac{1}{\sqrt{3^l}}I^{\ot l}\\
        W_2^{11}[l]&=\frac{1}{\sqrt{\lambda^l}}T^{\ot l}=\frac{1}{\sqrt{\lambda^l}}\bigoplus_{\eta\in {\rm spec}\bk{T^{\ot l}}} \eta\Pi_\eta,
    \end{align*}
    where ${\rm spec}\bk{T^{\ot l}}$ is the set of distinct eigenvalues of $T^{\ot l}$ and $\Pi_\eta$ is the projection onto the corresponding eigensubspace.  
    
    The reduced states of $\rho^{(L)}$ are completely mixed states, so one can directly calculate the minimal sufficient subalgebras $\cM_l^{\min}$ using the definition~\eqref{eq:defminsufalg}. A simple calculation shows that   
        \begin{equation}
        \cM_{l}^{\min}=\alg\{I^{\ot l}, T^{\ot l} \}\cong\mathbb{C}^{|{\rm spec}\bk{T^{\ot l}}|}\,.
    \end{equation}
    Since we have \begin{align}        \lr{| {\rm spec}\bk{T^{\ot l}}}|=\left(\begin{array}{c} l + 2 \\2 \end{array}         \right)=\frac{(l+2)(l+1)}{2}, \end{align} $\cM_l^{\min}$ is not supported on a constant dimensional Hilbert space. Therefore, these MPDOs cannot admit a converging exact renormalization flow.

\section{The structure of MPDOs with RG transformation}\label{sec:structure}
Because not all MPDOs admit an exact RG flow, here we study a subclass of MPDOs that has a well-behaved exact RG flow. 

Before defining the subclass, we introduce a crucial property of the minimal subalgebras of MPDOs. 
Divide the subregion $\Lambda_l$ to smaller segments with length $l_1, l_2$. The minimal sufficient subalgebra on $\Lambda_l$ is a subalgebra of the tensor product of the minimal sufficient subalgebras of the smaller regions.     
    \begin{lem}\label{lem:minalginclusion}
    Let $l_1+l_2=l$ and $K_{l_1,l_2}$ be the kernel of $\rho_{l_1+l_2}$ in $\supp(\rho_{l_1}\ot\rho_{l_2})$. Then the following holds:
    \begin{equation}
        \cM^{\min}_{l_1+l_2}\oplus0_{K_{l_1,l_2}}\subset\cM_{l_1}^{\min}\ot\cM_{l_2}^{\min}\,.
    \end{equation}
    \end{lem}
        We denote the inclusion map in this lemma by 
        \begin{equation*}
        \iota_{l_1+l_2}:     \cM^{\min}_{l_1+l_2}\oplus0_{K_{l_1,l_2}}\to\cM_{l_1}^{\min}\ot\cM_{l_2}^{\min}\,.
        \end{equation*}
        For simplicity, we will omit $0_{K_{l_1,l_2}}$ from the rest of the discussion.

    The subclass of MPDOs with an RG flow is then defined as follows:
    \begin{definition}\label{def:MPDORG}
        We say MPDOs generated by $M$ admit an exact RG flow if there is a finite-dimensional unital $C^*$-algebra 
        \begin{equation}
            \cA=\bigoplus_{a=1}^r\cB\bk{\C^{d_a}}
        \end{equation}
        and a family of injective representations $\{\pi_l\}_{l\in\mathbb{N}}$ such that
            \begin{equation}\label{eq:min=piA}
                \cM_l^{\min}=\pi_l(\cA)\,
            \end{equation}
        and there is a map $\Delta$ such that
\begin{equation*}
\xymatrix{
\cA\ar[r]^-{\Delta}\ar[d]_-{\pi_{l_1+l_2}}\ar@{}[rd]|{\circlearrowright}&\cA\ot \cA\\
\pi_{l_1+l_2}(\cA)\ar@{^{(}->}[r]_-{\iota_{l_1+l_2}}&\pi_{l_1}(\cA)\ot\pi_{l_2}(\cA)\ar[u]_-{\pi_{l_1}^{-1}\ot\pi_{l_2}^{-1}}
}
\end{equation*}
is a commutative diagram for any $l_1,l_2\in\mathbb{N}$. 
    \end{definition}
The first condition~\eqref{eq:min=piA} is almost identical to the necessary condition~\eqref{eq:anecessary}. More generally, an infinite sequence of finite-dimensional $C^*$-algebras $\{\cA_l\}_{l\geq l_0}$ may be considered. 
However, we do not investigate this possibility here. The existence of the map  $\Delta$ implies that 
\begin{equation*}
    \Delta_{l_1+l_2}:=\bk{\pi_{l_1}^{-1}\ot\pi_{l_2}^{-1}}\circ\iota_{l_1+l_2}\circ\pi_{l_1+l_2}
\end{equation*}
is independent of the choice of $l_1,l_2$. This might follow from the translation invariance of the MPDO, but because we do not have a proof, we assume it. 

A simple but important consequence of these conditions is the following:
\begin{prop}\label{prop:comultiplication}
$\Delta:\cA\mapsto \cA\ot\cA$ is a multiplicative co-multiplication on $\cA$, i.e., it satisfies 
    \begin{equation*}
        \Delta(xy)=\Delta(x)\Delta(y) \quad \forall x,y\in\cA\,
    \end{equation*}
    and
\begin{equation*}
    \bk{\id\ot\Delta}\circ\Delta=\bk{\Delta\ot \id}\circ\Delta\,=:\Delta^2\,.
\end{equation*}
\end{prop}
From the assumption, we have in particular, 
\begin{equation*}
    \iota_{1+1}\circ\pi_{2}=\bk{\pi_1\ot\pi_1}\circ\Delta\,.
\end{equation*}
Let us simply denote $\pi_1$ by $\pi$ and omit the inclusion. Then, recursively applying this relation we obtain
\begin{equation}
    \pi_l=\pi^{\ot l}\circ\Delta^{L-1}\,.
\end{equation}

$\cA$ is called a {\it coalgebra} if there is a comultiplication and a counit $\epsilon:\cA\to\C$ satisfying
\begin{equation*}
    \bk{\id\otimes \epsilon}\circ\Delta=\bk{\epsilon\otimes\id}\circ\Delta=\id\,.
\end{equation*}
An algebra with a multiplicative coalgebra structure is called a {\it pre-bialgebra}.
\begin{definition}
    An algebra $\cA$ is a pre-bialgebra if it is also a coalgebra and its comultiplication $\Delta$ is multiplicative.

\end{definition}
We show that the algebra $\cA$ associated with the RG flow not only has a multiplicative comultiplication (Proposition~\ref{prop:comultiplication}), but also has a counit and is always a prebialgebra (the proof can be found in Appendix~\ref{app:prebialgebra}).
\begin{thm}\label{thm:prebialgbra}
    If MPDO $\rho^{(L)}$ admits an exact RG flow for $(\pi_l,\cA)$, then $\cA$ is a pre-bialgebra\,.
\end{thm}
The pre-bialgebra structure is crucial for characterizing the MPO algebra, which is an algebra consisting of MPOs with a fixed bond dimension~\cite{MPOalg1,MPOalg2}. An important subclass of pre-bialgebra is the class of weak Hopf algebras. The categories of modules of weak Hopf algebras are equivalent to fusion categories~\cite{etingof2017fusioncategories}, which classify 2D topologically ordered phases.

One of the main results of this paper is a structure theorem on MPDOs with an RG flow. 
\begin{thm}\label{thm:MPDOstructure}
    For any MPDO $\rho^{(L)}$ with an exact RG flow, there exists an element  $w^{(L)}\in \cA$ such that
    \begin{equation}
       \rho^{(L)}=\pi^{\ot L}\circ\Delta^{L-1}\bk{w^{(L)}}\Omega^{(L)}\,,
    \end{equation}
    where $\Omega^{(L)}>0$ satisfies $[\pi^{\ot L}\circ\Delta^{L-1}\bk{\cA},\Omega^{(L)}]=0$ and has the structure
    \begin{equation}
        \Omega^{(L)}\cong\bigoplus_{c=1}^r I_{d_c}\ot\Omega_c^{(L)}
    \end{equation}
    where
    \begin{align*}
     \Omega_c^{(L)}&:=\bigoplus_{{\bf a,b}}\Gamma_{1,a_1b_1}^{c}\ot\Gamma_{2,a_2b_2}^{b_1}\ot\cdots\ot \Gamma_{(L-1),a_{L-1}a_L}^{b_{L-2}}\bk{\bigotimes_{k=1}^L \Omega_{a_k}}\,,
    \end{align*}
    with diagonal matrices $\Gamma_{k,ab}^{c}$ $(a,b,c=1,...,r,k=1,...,L-1)$ where each is either  $\Gamma_{k,ab}^{c}>0$ or  $\Gamma_{k,ab}^{c}=0$. 
\end{thm}
The proof is given in Appendix~\ref{app:structurethm}. 
 \begin{figure}[htbp]
	\centering
	\includegraphics[width=0.4\textwidth]{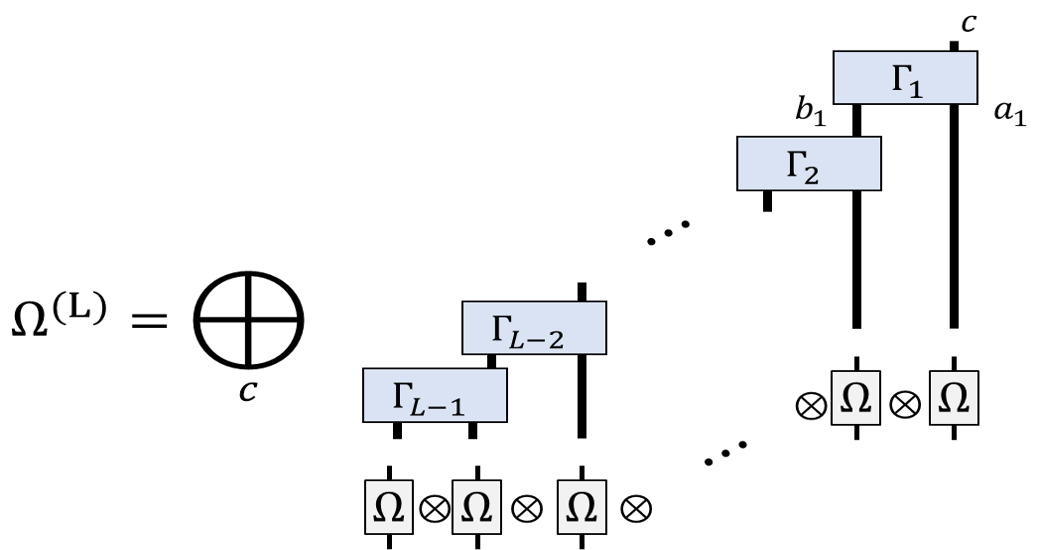}
	\caption{A diagrammatic structure of $\Omega^{(L)}$, where $\Omega:=\bigoplus_a\Omega_a$. It consists of a tree-like tensor of $\Gamma_{k,a,b}^c$ and a product state $\Omega^{\ot L}$, and they commute with each other. In Ref.~\cite{CIRAC2017100}, $\Omega^{\ot L}$ is treated as the local Gibbs state and $\Gamma_{k,a,b}^c$ is absorbed to the global MPO operators, while we consider that $\Gamma_{k,a,b}^c$ is also a part of locally interacting (non-topological) part. This makes the non-topological part truly interacting in our structure theorem.}
	\label{fig:Omega}
\end{figure}

This theorem generalizes the structure theorem for the fixed point~\eqref{eq:RFPMPDO} in two ways. First, it applies to non-fixed-point MPDOs. Second, in our theorem, the global MPO (``topological part" $\bigoplus\lambda_iP_i$) is explicitly identified as an MPO representation of an element of a pre-bialgebra, $\pi^{\ot L}\circ\Delta^{L-1}\bk{w^{(L)}}$. In Ref.~\cite{CIRAC2017100}, $\Gamma_{k,a,b}^c$ is also regarded as part of the topological component together with $\pi^{\ot L}\circ\Delta^{L-1}\bk{w^{(L)}}$ in our notation. This requires each $\Gamma_{k,a,b,}^c$ to be a projector, and the locally interacting part to be non-interacting ${\bigoplus_a\Omega_a}^{\otimes L}$, represented as in Eq.~\eqref{eq:RFPMPDO}~\cite{private}. In our case, we regard $\Gamma_{k,a,b}^c$ as a non-topological part ($\Omega^{(L)}$). Because of this identification, we can accommodate  interacting states, such as the Gibbs state of a commuting local Hamiltonian, as the ``non-topological part". Although we do not have a proof, we expect to $\Omega^{(L)}$ always be ``locally interacting"-- for example, expressible as the Gibbs state of a commuting local Hamiltonian. We can consider the following RG fixed point MPDOs to support this conjecture.

\begin{remark}
If $\cA=\mathbb{C}$, then $\cM^{\min}_l=\mathbb{C}I$. This implies that:
\begin{equation}
    \forall \mu\in \cS, \exists c\in\mathbb{C} \; s.t. \;\mu\rho^{-1}_{l}=cI\Rightarrow \mu=\rho_{l}
\end{equation}
and therefore $\rho^{(L)}=\rho^{\ot L/l}_l$. From the translation invariance, we have $\rho^{(L)}=\Omega^{(L)}=\rho^{\ot L}_1$, i.e., MPDO is a product state.
\end{remark}

\begin{remark}
Let $\sigma\in\cS(\C^d\otimes\C^d)$ be a positive bipartite mixed state. Let 
\begin{equation}
    \sigma=\sum_{\alpha=1}^D F^\alpha\otimes G^{\alpha} 
\end{equation}
be the operator Schmidt decomposition. Define a tensor $M=\sum W^{\alpha\beta}\ot \ketr{\alpha}\brar{\beta}$ by
\begin{equation}
    W^{\alpha\beta}:=G^\alpha\ot \sigma \ot F^\beta\,.
\end{equation}
Each physical site associated with a single $M$ is $\bk{\C^d}^{\ot 4}$. 
The MPDO generated by $M$ is a product of mixed dimers $\sigma$, and interacts with neighboring sites. The h-CF only contains one block.

We label $\C^d$ from the left by $i=1,2,3,...$. The v-CF of $W^{\alpha\beta}[l]$ is then:
\begin{equation*}
    W^{\alpha\beta}[l]\cong (G^\alpha\ot F^\beta) \bk{\bigotimes_{i=1}^{l-1}\sigma_{4i,4i+1}}\bk{\bigotimes_{j=1}^l\sigma_{4j-2,4j-1}}
\end{equation*}
from which we can identify 
\begin{equation*}
    \Omega^{(L)}=\bk{\bigotimes_{i=1}^{l-1}\sigma_{4i,4i+1}}\bk{\bigotimes_{j=1}^L\sigma_{4j-2,4j-1}}\,,
\end{equation*}
where 
\begin{equation*}
    \bigoplus_{{\bf a,b}}\Gamma_{1,a_1b_1}^{c}\ot\Gamma_{2,a_2b_2}^{b_1}\ot\cdots\ot \Gamma_{(L-1),a_{L-1}a_L}^{b_{L-2}}=\bk{\bigotimes_{i=1}^{L-1}\sigma_{4i,4i+1}}\,.
\end{equation*}
$\Omega^{(L)}$ is a Gibbs state $e^{-H_L}$ of a commuting local Hamiltonian $H_L$, whose interaction terms are given by
\begin{equation*}
    H_L:=-\sum_{i=1}^L\ln\sigma_{4i,4i+1}-\sum_{i=1}^{L}\ln\sigma_{4i-2,4i-1}\,.
\end{equation*}
\end{remark}

\subsection{The exact RG-transformation and its inverse}
In this section we show one construction of an exact RG flow and its inverse. 
For a given MPDO satisfying the conditions in Definition~\ref{def:MPDORG}, the first step of the exact RG transformation is to apply the unitary transformation the $l$-consecutive tensor to the v-CF~\eqref{eq:v-CF} 
        \begin{align*}
            U_lW_k^{\alpha_k\beta_k}[l]U_l^\dagger=\bigoplus_{a=1}^{r} W_{k,a}^{\alpha_k\beta_k}[l] \ot \Omega_{a}^{(l)}\,.
        \end{align*}
Each $W_{k,a}^{\alpha_k\beta_k}[l]$ is supported on a $d_a$-dimensional Hilbert space $\C^{d_a}$ for any $l$, due to the condition~\eqref{eq:min=piA}. $\Omega_{a}^{(l)}$ is supported on a large (say, $d_a^R[l]$-dimensional) Hilbert space $\cH_{d_a^R[l]}$ and carries local correlations and entropy irrelevant to the correlation between the blocked sites. Similarly to the Koashi-Imoto decomposition~\eqref{eq:KI-dec}, can be discard $\Omega_{a}^{(l)}$ without losing any relevant correlation. 

Let $\Pi_a$ be the projection onto $\cH_{d_a}\ot\cH_{d_a^R[l]}$. Then the second step of the renormalization transformation is
\begin{align*}
    E_{l\to1}(X):=\bigoplus_{a=1}^{r} \tr_{\cH_{d_a^R[l]}}\bk{\Pi_aX\Pi_a}\ot \frac{\Omega_a^{(1)}}{\tr{\bk{\Omega_a^{(1)}}}}    \,,
   % W_{k,a}^{\alpha_k\beta_k}[l] \ot \Omega_{a}^{(l)}
   % &\mapsto  \bk{\bigoplus_{a=1}^{r}\frac{\tr\bk{\Omega_a^{(l)}}}{\tr\bk{\Omega_a^{(1)}}}W_{k,a}^{\alpha_k\beta_k}[l] \ot \Omega_a^{(1)}}U_1\,\\
  %  &\equiv{\tilde W}^{\alpha_k\beta_k}_k[1].
\end{align*}
which is a quantum channel from $\cB\bk{(\C^d)^{\ot l}}$ to $\cB\bk{\C^d}$. Finally, we return to the original basis of a single site using $U_1$. In summary, the renormalization transformation on $l$-neighboring sites is 
\begin{equation*}
    \cE_{l\to1}(\cdot):=U_1^\dagger E_{l\to 1}\bk{U_l\cdot U_l^\dagger}U_1\,.
\end{equation*}
After renormalization, the new single site tensor ${\tilde M}$ is defined as 
\begin{align}
    \tilde{M}&=\sum_{k}\sum_{\alpha_k,\beta_k=1} \tilde{W}_k^{\alpha_k\beta_k}\ot\ketr{\alpha_k}\brar{\beta_k}\ot \tilde{N}_k\,,\label{eq:RGMtilde}\\
    \tilde{W}_k^{\alpha_k\beta_k}&:=U_1^\dagger\bk{\bigoplus_{a=1}^r \frac{\tr\bk{\Omega_{a}^{(l)}}}{K_{k,l}\tr\bk{\Omega_{a}^{(1)}}}W^{\alpha_k\beta_k}_{k,a}[l]\ot\Omega^{(1)}_{a}}U_1\,,\nonumber\\
    \tilde{N}_k&:=K_{k,l}N_k^l\,,\nonumber
\end{align}
where $K_{k,l}$ is a constant to normalize 
\begin{equation*}
    \lr{\{\frac{\tr\bk{\Omega_{a}^{(l)}}}{K_{k,l}\tr\bk{\Omega_{a}^{(1)}}}W^{\alpha_k\beta_k}_{k,a}[l]}\}
\end{equation*} 
to form a BNT.

The inverse of the renormalization transformation is constructed using the recovery map of $E_{l\to 1}$, defined as:
\begin{equation*}
    F_{1\to l}(X):=\bigoplus_{a=1}^{r} \tr_{\cH_{d_a^R[1]}}\bk{\Pi_aX\Pi_a}\ot \frac{\Omega_a^{(l)}}{\tr{\bk{\Omega_a^{(l)}}}}\,,
\end{equation*}
and written as:
\begin{equation*}
    \cF_{1\to l}(\cdot):=U_l^\dagger F_{1\to l}(U_1 \cdot U_1^\dagger)U_l\,.
\end{equation*}
One can explicitly check that 
\begin{equation*}
    W_k^{\alpha_k\beta_k}[l]=\cF_{1\to l}\circ\cE_{l\to1}(W_k^{\alpha_k\beta_k}[l])\,,\quad\forall k,\alpha_k,\beta_k\,
\end{equation*}
which guarantees that for any $L\geq l+1$, the MPDO $\rho^{(L)}_{AB}$ with a partition $AB=\Lambda_L$ also satisfies:
\begin{equation*}
    \rho^{(L)}_{AB}=\cF_{1\to B}\circ\cE_{B\to1}\bk{\rho^{(L)}_{AB}}\,.
\end{equation*}

The entire exact RG flow is constructed by layers of parallel applications of $\cE_{l\to 1}$ for different $l$. Let us fix the coarse-graining scale to $l_0=L/L_2$ for some $L_2$. The first step of RG flow is given by $\cE_{l_0\to 1}^{\ot L_2}$, which reduces the number of sites to $L_2$ while keeping the local dimension constant at $d$. 
The second step is again coarse-graining $l_0$-neighboring (coarse-grained) sites into a single site. This coarse-graining operation is constructed using the v-CF or the minimal sufficient subalgebra on the $l_0$ coarse-grained sites. However, the minimal sufficient subalgebra changes isomorphically under reversible quantum channels~\cite{10.1093/qmath/39.1.97,Jencova2006-on}; therefore, one can invert the first step and then apply coarse-graining on the $l_0^2$ sites instead. Thus, the second renormalization step is thus given by:
\begin{equation*}
    \cE^{(2)}_{l_0\to 1}:=\cE_{l_0^2\to1}\circ\cF_{1\to l_0}^{\ot l_0}\,.
\end{equation*}
Note that $\cE^{(2)}_{l_0\to 1}$, which is a quantum channel from $l_0$-sites to a single site, may not necessarily be implemented by the above sequential representation. From this, we see that level-$k$ of the renormalization is equivalent to applying $\cE_{l_0^k\to1}$ to the original tensor, i.e.,
\begin{equation*}
   \cE^{(k)}_{l_0\to 1}\circ\bk{\cE^{(k-1)}_{l_0\to 1}}^{\ot l_0}\circ\cdots\circ\bk{\cE_{l_0\to 1}}^{\ot l_0^{k-1}}( \rho^{(L)})=\cE_{l_0^k\to1}( \rho^{(L)})
\end{equation*}
for any $L\geq l_0^k$, where both maps act on neighboring $l_0^k$ sites.

We remark on alternative choices for the renormalization transformation. Although our construction is not the only option, the existence of an inverse operation imposes a significant constraint on the exact renormalization map. According to the Koashi-Imoto theorem, any change in $W_{k,a}^{\alpha_k\beta_k}[l]$ or in the direct sum structure results in information loss. Thus any exact renormalization transformation must preserve this structure (up to a unitary). Therefore, other possible transformations can only differ by  changes in $\Omega_a^{l}$ or an isometry.  Furthermore, by definition, these exact renormalization transformations are always locally compatible with other CPTP maps. 

\subsection{The fixed points of RG flow}
When the RG flow described in the previous section converges, the fixed points are characterized by 
\begin{equation*}
    {\tilde M}=M\,
\end{equation*}
for $\tilde M$ in Eq.~\eqref{eq:RGMtilde}. 
This condition is met when there exists $c_k\in\C, (k=1,...,g)$ satisfying
\begin{align*}
    N_k^l&=c^{l-1}_kN_k\,,\\
    \tr\bk{\Omega_{a}^{(l)}}W_{k,a}^{\alpha_k\beta_k}[l]&=c_k^{1-l}\tr\bk{\Omega_{a}^{(1)}}W_{k,a}^{\alpha_k\beta_k}
\end{align*}
for any $l$. 
We will present an explicit example in Sec.~\ref{sec:exampleX}. 

As mentioned in Sec.~\ref{sec:RFP1}, these fixed points were first introduced and studied in Ref.~\cite{CIRAC2017100}. A subclass of fixed points constructed from a biconnected $C^*$-weak Hopf algebra,  is discussed in Ref.~\cite{MPOalg2}. Those fixed point MPDOs can be understood as states in boundary theories of 2D tensor networks in topologically ordered phases.

\section{Example of MPDO with a RG flow}\label{sec:exampleX}
In this section, we present an explicit example of an MPDO with an RG flow that is not a fixed point. This is a classical state, and a similar model is also studied in Ref.~\cite{mixedphase1}.

Let $0<p\leq 1$ and consider a MPDO
    \begin{equation*}
        \rho^{(L)}_p=\frac{1}{2^L}\left(I^{\ot L}+p^lX^{\ot L}\right)\,,
    \end{equation*}
    where $X$ is the Pauli $X$ matrix. 
    The corresponding tensor is defined as 
    \begin{equation*}
        M^{ij}=\frac{\delta_{ij}}{2}\ketr{0}\brar{0}+\frac{p\delta_{i{\bar j}}}{2}\ketr{1}\brar{1} \quad(i,j=0,1)\,,
    \end{equation*}
    where ${\bar j}=j+1$ (mod $2$). 
  The h-CF of this tensor is given as  
    \begin{align*}
        M^{ij}&\cong(M_0^{ij}\otimes N_0)\oplus(M_1^{ij}\otimes N_1)\,,\\
        M_0^{ij}\otimes N_0&=\left(\frac{\delta_{ij}}{\sqrt{2}}\ketr{0}\brar{0}\right)\ot\left(\frac{1}{\sqrt{2}}\ketr{0}\brar{0}\right)\,,\\
        M_1^{ij}\otimes N_1&=\left(\frac{\delta_{i{\bar j}}}{\sqrt{2}}\ketr{1}\brar{1}\right)\ot\left(\frac{p}{\sqrt{2}}\ketr{1}\brar{1}\right)\,.
    \end{align*}
By blocking $l$ tensors,
    \begin{align*}
        &\overbrace{M^{i_1j_1}...M^{i_lj_l}}^{l}\\
        &=\lr{[\left(\frac{\prod_{k=1}^l\delta_{i_kj_k}}{\sqrt{2^l}}\ketr{0}\brar{0}\right)\ot N_0^l}]\oplus \lr{[\left(\frac{\prod_{k=1}^l\delta_{i_k{\bar j_k}}}{\sqrt{2}}\ketr{1}\brar{1}\right)\ot N_1^l}]\,.
    \end{align*}
    
We obtain
\begin{align*}
    W_0^{11}[l]=\frac{1}{\sqrt{2^l}}I^{\ot l}\,,\quad W_1^{11}[l]=\frac{1}{\sqrt{2^l}}X^{\ot l}\,
\end{align*}
($D_0^L=D_1^L=1$, therefore the only meaningful case is $\alpha_k,\beta_k=1$). We will omit upper indices $^{11}$ from this point onward. The v-CF for length $l$ is therefore given by 
\begin{align*}
    W_0[l]&\cong \lr{[\frac{1}{\sqrt{2}}\ket{0}\bra{0}\ot\frac{1}{\sqrt{2^{l-1}}}\Pi_{X^{\ot l}=+1} }]\\
    &\quad\oplus\lr{[\frac{1}{\sqrt{2}}\ket{1}\bra{1}\ot\frac{1}{\sqrt{2^{l-1}}}\Pi_{X^{\ot l}=-1} }] \,,\\
    W_1[l]&\cong \lr{[\frac{1}{\sqrt{2}}\ket{0}\bra{0}\ot\frac{1}{\sqrt{2^{l-1}}}\Pi_{X^{\ot l}=+1} }]\\
    &\quad\oplus\lr{[\frac{-1}{\sqrt{2}}\ket{1}\bra{1}\ot\frac{1}{\sqrt{2^{l-1}}}\Pi_{X^{\ot l}=-1} }]\,,
\end{align*}
where $\Pi_{X^{\ot l}=\pm1}$ are the projections of the eigensubspaces of $X^{\ot l}$. The coefficients are determined by the BNT condition.

Let us define the operators $w_0[l], w_1[l]$ and $\Omega_l$ as  
\begin{align*}
    w_0[l]&\:=\ket{0}\bra{0}\oplus\ket{1}\bra{1}\,,\\
    w_1[l]&\:=\ket{0}\bra{0}\oplus\bk{-\ket{1}\bra{1}}\,,\\
    \Omega^{(l)}&:=\lr{[\ket{0}\bra{0}\ot\frac{1}{\sqrt{2^{l}}}\Pi_{X^{\ot l}=+1} }]\oplus\lr{[\ket{1}\bra{1}\ot\frac{1}{\sqrt{2^{l}}}\Pi_{X^{\ot l}=-1} }]\\
    &\,\cong \frac{1}{\sqrt{2^{l}}}I^{\ot l}.
\end{align*}

The corresponding algebra $\cA$ is 
\begin{equation*}
    \cA=\lr{\{ c_0\ket{0}\bra{0}\oplus c_1\ket{1}\bra{1} \:\middle|\: c_0,c_1\in\mathbb{C} }\}
\end{equation*}
and the representation $\pi_l$ is defined as 
\begin{equation*}
    \pi_l(\cA)=\lr{\{ c_0\ket{0}\bra{0}\ot
    \Pi_{X^{\ot l}=+1}\oplus c_1\ket{1}\bra{1} \ot\Pi_{X^{\ot l}=-1}}\}\,.
\end{equation*}
Since the projections on to the eigensubspaces can be decomposed as
\begin{align*}
     \Pi_{X^{\ot2 l}=+1}&= \Pi_{X^{\ot l}=+1}\ot\Pi_{X^{\ot l}=+1}+ \Pi_{X^{\ot l}=-1}\ot\Pi_{X^{\ot l}=-1}\,,\\
     \Pi_{X^{\ot2 l}=-1}&= \Pi_{X^{\ot l}=+1}\ot\Pi_{X^{\ot l}=-1}+ \Pi_{X^{\ot l}=-1}\ot\Pi_{X^{\ot l}=+1}\,,
\end{align*}
the inclusion relation is explicitly written as
\begin{align*}
   \iota_{l+l}:\: &c_0\ket{0}\bra{0}\ot
    \Pi_{X^{\ot 2l}=+1}+c_1\ket{1}\bra{1} \ot\Pi_{X^{\ot 2l}=-1}\in\pi_{2l}(\cA)
   \\\: &\mapsto \\
    & c_0\bk{\ket{0}\bra{0}\ot
    \Pi_{X^{\ot l}=+1}}\ot\bk{\ket{0}\bra{0}\ot
    \Pi_{X^{\ot l}=+1}}\\
    &+c_0\bk{\ket{1}\bra{1}\ot
    \Pi_{X^{\ot l}=-1}}\ot\bk{\ket{1}\bra{1}\ot
    \Pi_{X^{\ot l}=-1}}\\
    &+c_1\bk{\ket{0}\bra{0}\ot
    \Pi_{X^{\ot l}=+1}}\ot\bk{\ket{1}\bra{1}\ot
    \Pi_{X^{\ot l}=-1}}\\
    &+c_1\bk{\ket{1}\bra{1}\ot
    \Pi_{X^{\ot l}=-1}}\ot\bk{\ket{0}\bra{0}\ot
    \Pi_{X^{\ot l}=+1}}\,\\
    &\in\pi_l(\cA)\ot\pi_l(\cA).
\end{align*}
In conclusion, the coproduct is defined as 
\begin{align*}
    \Delta(c)&:=\bk{\pi_l^{-1}\ot\pi_l^{-1}}\circ\iota_{l+l}\circ\pi_{2l}(c)\\
    &\:=c_0 \bk{\ket{0}\bra{0}\ot\ket{0}\bra{0}\oplus\ket{1}\bra{1}\ot\ket{1}\bra{1}}\\
    &\quad\oplus c_1 \bk{\ket{0}\bra{0}\ot\ket{1}\bra{1}\oplus\ket{1}\bra{1}\ot\ket{0}\bra{0}}
\end{align*}
for any $c=c_0\ket{0}\bra{0}\oplus c_1\ket{1}\bra{1}\in \cM$. One can easily verify that $\Delta$ is independent of $l$ (more generally, $\Delta_{l_1+l_2}=\Delta$ also holds).   The basis of $\cA$ is given by $\ket{0}\bra{0},\ket{1}\bra{1}$ and the dual basis in $\cA^*$ is given by $\delta_+:=\tr\bk{\ket{0}\bra{0}\cdot}, \delta_-:=\tr\bk{\ket{1}\bra{1}\cdot}$. 
The multiplication of $\cA^*$ is defined via
\begin{equation*}
    \delta_x\delta_y:=\left(\delta_x\ot\delta_y\right)\circ\Delta\,.
\end{equation*}
From this definition, one can confirm that 
\begin{equation*}
    \delta_x\delta_y(c)=\begin{cases}
c_0=\delta_+(c)& x=y \\
c_1=\delta_-(c) & x\neq y\,
\end{cases}
\end{equation*}
and therefore 
\begin{equation*}
    \delta_+(a\delta_++b\delta_-)=(a\delta_++b\delta_-)\delta_+=(a\delta_++b\delta_-)\,, \forall a,b\in\mathbb{C}\,,
\end{equation*}
i.e., $\delta_+$ is the unit of $\cA^*$ (the counit of $\cA$). By definition, $\Delta$ is multiplicative, and thus, $\cA$ and $\cA^*$ are pre-bialgebras. Moreover, this example satisfies $\Delta(1)=1\ot 1$.

Meanwhile,
\begin{align*}
    \psi(\delta_\pm)&:=\delta_\pm\bk{w_0[l]}\ket{0}\bra{0}\ot I_{D_0^R}\oplus\delta_\pm\bk{w_1[l]}\ket{1}\bra{1}\ot I_{D_1^R}\,\\
    &\:=\ket{0}\bra{0}\ot \ket{0}\bra{0}\oplus\bk{\pm\ket{1}\bra{1}}\ot \ket{1}\bra{1}\,,
\end{align*}
and this provides an injective representation of $\cA^*$.

The MPDO is then
\begin{align*}
    \rho^{(L)}_p&=\frac{1}{\sqrt{2^L}}\pi^{\ot L}\circ\Delta\bk{w^{(L)}}\\
    w^{(L)}&=\frac{1+p^L}{\sqrt{2^L}}\ket{0}\bra{0}+\frac{1-p^L}{\sqrt{2^L}}\ket{1}\bra{1}\,.
\end{align*}
The RG flow converges to (Fig.~\ref{fig:phaseX})
\begin{align}
    \begin{cases}
       \frac{1}{2^L}I^{\ot L} & 0\leq p<1\\
       \frac{1}{2^{L-1}}\Pi_{X^{\ot L}=+1} & p=1\,.
    \end{cases}
\end{align}

 \begin{figure}[htbp]
	\centering
	\includegraphics[width=0.4\textwidth]{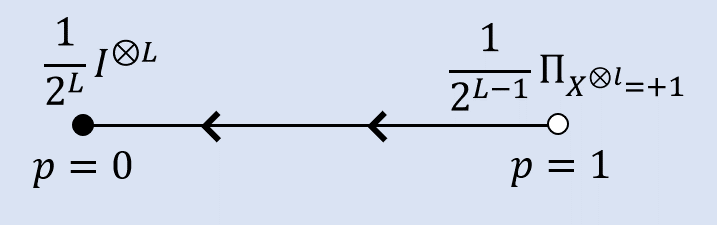}
	\caption{The exact RG flow of $\rho^{(L)}$. Any state with $p<1$ converges to the completely mixed state, while only the $p=1$ state represents another fixed point and remains unchanged. Although there are two distinct fixed points, all MPDOs in $0\leq p\leq1$ belong to the same mixed-state quantum phase~\cite{MPOalg2}.}
	\label{fig:phaseX}
\end{figure}

\section{MPO-algebra symmetry of MPDOs with RG flows}\label{sec:symmetry}
Traditionally, the symmetry of a physical system is described by a group and its action on the system (representation). Group symmetry is a powerful tool for understanding phases of matter through the concept of spontaneous symmetry-breaking. However, certain phases known as topologically ordered phases, cannot be characterized within this group symmetry framework. The symmetries of topologically ordered phases are non-invertible symmetries (categorical symmetries) from algebraic structures, which extend beyond traditional group symmetry. 

For a topologically ordered 2D tensor network state, a non-invertible symmetry is represented by MPOs acting on the virtual degrees of freedom~\cite{MPOinjectivePEPS}. This imposes an MPO-algebra symmetry on the boundary theory, where the state is typically an MPDO~\cite{PEPSboundary,PEPSboundary0,Kastoryano_2019}. 

A straightforward corollary of Theorem~\ref{thm:MPDOstructure} is that any MPDO with an RG flow must obey an abelian MPO-symmetry characterized by the center of the pre-bialgebra. 
\begin{cor}\label{thm:MPOsym}
For any MPDO $\rho^{(L)}$ with an exact RG flow, let $\cZ(\cA)$ be the center of the pre-bialgebra $\cA$. Then, for any $z\in\cZ(\cA)$, there exists an MPO representation of $\cA$ 
\begin{equation*}
    \pi^{\ot L}\circ\Delta^{L-1}\bk{z}=\sum_{\bf i,j} \tr\bk{b_zT^{i_1j_1}...T^{i_Lj_L}}\ket{i_1...i_L}\bra{j_1...j_L}\,,
\end{equation*}
where $b_z$ and $T^{ij}$ are fixed-dimensional matrices, such that
\begin{equation*}
    \lr{[\rho^{(L)}, \pi^{\ot L}\circ\Delta^{L-1}\bk{z}}]=0\,,\quad \forall L\,.
\end{equation*}
\end{cor}
Note that this corollary holds for any MPDO with RG flow. Specifically, when the MPDO is the boundary state of a 2D topologically ordered ground state, it is expected that the pre-bialgebra is a weak-Hopf algebra, and its representation category reproduces the fusion category of the topological order.

\section{Classification of 1D mixed-state quantum phases}\label{sec:classification}
Classifying phases of matter is one of the fundamental problems in many-body physics. In the most general sense, two gapped local Hamiltonians are said to be in the same (gapped) quantum phase if they are connected by a smooth deformation of local interactions without closing the energy gap. Phases that do not admit a trivial product state are known as topologically ordered phases. Furthermore, a finer classification by, e.g., imposing symmetries along the path, leading to the classification of SPT phases and symmetry-enriched topologically ordered phases. 

 Two gapped local Hamiltonians are connected by a smooth gapped path if and only if their ground states are (approximately) connected by a short-depth local unitary circuit~\cite{LU1,LU2}. Therefore, quantum phases are associated with equivalence classes of ground states connected by unitary circuits. A natural extension of this concept to finite-temperature, or more generally, to mixed states is achieved by replacing the local unitary circuits with local circuits of general quantum channels. The equivalence classes of mixed states are referred to as mixed-state quantum phases~\cite{PhysRevLett.107.210501,mixedphase3,mixedphase1}. 

The framework of real-space, exact RG transformations is a powerful tool for characterizing  equivalence classes. When the RG flow converges to a fixed point, the flow become an explicit local circuit of unitary or quantum channels connecting the starting point to the fixed point. Since the RG flow is invertible, the starting point and the fixed point must lie within the same quantum phase. Therefore, to classify quantum phases, it is enough to classify the fixed points of the RG flows. 

The fixed points rely on the pre-bialgebra structure, and it is expected that mixed-state quantum phases are classified by pre-bialgebras. In Ref.~\cite{MPOalg2}, it is shown that if the pre-bialgebra satisfies additional conditions and become a $C^*$-Hopf algebra, then the MPDO is in the trivial phase (connected to a product state). The example shown in Sec.~\ref{sec:exampleX} possesses non-trivial MPO symmetry, but the pre-bialgebra $\cA\cong \C^2$ is actually a $C^*$-Hopf algebra, and therefore the MPDO is in the trivial phase. It is conjectured that at least one non-trivial class of 1D mixed-state quantum phase exists within MPDOs~\cite{MPOalg2}, but providing this remains an open problem. 

\section{Conclusion and discussions}\label{sec:conclusion}
In this paper, we investigated real-space exact RG transformations of MPDOs, which consist of short-depth circuits of local quantum channels, requiring the renormalization flow to remain invertible. A key distinction was drawn between MPDOs and MPSs, highlighting that while MPSs always admit well-defined isometric renormalization flows, MPDOs do not always exhibit a converging renormalization flow.

To address this, we introduced a subclass of MPDOs that possess a well-defined renormalization flow, which we analyzed in the context of an MPO representation of a pre-bialgebra structure. This led to the observation that MPDOs in this subclass obey generalized symmetries governed by MPO algebras. Additionally, we examined the fixed points of the renormalization flow, providing insights into the classification of mixed-state quantum phases. Our result partially addresses the question of what kinds of physical states are equivalent to all MPDOs. We conjecture that a general MPDO can be expressed as a product of an MPO-algebra, representing the anomalous part from the boundary theory, and a local Gibbs state. An interesting open question is whether there exists an MPDO that obeys a non-trivial MPO symmetry but cannot be realized by any boundary theory of 2D topologically ordered phases.

To obtain a complete classification of 1D mixed-state quantum phases, the fixed points must first be classified. Our structure theorem suggests that these equivalence classes are labeled by certain classes of pre-bialgebras, though the explicit correspondence remains an open question. We also revealed that, to apply the RG flow strategy to classification, one must consider  an approximate RG flow. This type of RG flow was developed for 2D tensor networks, such as the tensor renormalization group (TRG) method ~\cite{TRG1,TRG2}, b but it cannot be directly applied to 1D mixed states. The approximate invertibility of quantum channels is a mathematically challenging problem, and further analysis is required to determine which degrees of freedom can be truncated without altering the quantum phase.

\section*{Acknowledgement}
The author would like to thank David P\'{e}rez-Garcia, Alberto Ruiz de Alarc\'{o}n and Jos\'{e} Garre Rubio for their helpful comments and fruitful discussions. K. K. acknowledges support from the JSPS Grant-in-Aid for Early-Career Scientists, No. 22K13972; from the MEXT-JSPS Grant-in-Aid for Transformative Research Areas (A) “Extreme Universe,” No. 22H05254; and from the MEXT-JSPS Grant-in-Aid for Transformative Research Areas (B) No. 24H00829.

\bibliography{ref_MPDORG}

\begin{thebibliography}{49}%
\makeatletter
\providecommand \@ifxundefined [1]{%
 \@ifx{#1\undefined}
}%
\providecommand \@ifnum [1]{%
 \ifnum #1\expandafter \@firstoftwo
 \else \expandafter \@secondoftwo
 \fi
}%
\providecommand \@ifx [1]{%
 \ifx #1\expandafter \@firstoftwo
 \else \expandafter \@secondoftwo
 \fi
}%
\providecommand \natexlab [1]{#1}%
\providecommand \enquote  [1]{``#1''}%
\providecommand \bibnamefont  [1]{#1}%
\providecommand \bibfnamefont [1]{#1}%
\providecommand \citenamefont [1]{#1}%
\providecommand \href@noop [0]{\@secondoftwo}%
\providecommand \href [0]{\begingroup \@sanitize@url \@href}%
\providecommand \@href[1]{\@@startlink{#1}\@@href}%
\providecommand \@@href[1]{\endgroup#1\@@endlink}%
\providecommand \@sanitize@url [0]{\catcode `\\12\catcode `\$12\catcode `\&12\catcode `\#12\catcode `\^12\catcode `\_12\catcode `\%12\relax}%
\providecommand \@@startlink[1]{}%
\providecommand \@@endlink[0]{}%
\providecommand \url  [0]{\begingroup\@sanitize@url \@url }%
\providecommand \@url [1]{\endgroup\@href {#1}{\urlprefix }}%
\providecommand \urlprefix  [0]{URL }%
\providecommand \Eprint [0]{\href }%
\providecommand \doibase [0]{http://dx.doi.org/}%
\providecommand \selectlanguage [0]{\@gobble}%
\providecommand \bibinfo  [0]{\@secondoftwo}%
\providecommand \bibfield  [0]{\@secondoftwo}%
\providecommand \translation [1]{[#1]}%
\providecommand \BibitemOpen [0]{}%
\providecommand \bibitemStop [0]{}%
\providecommand \bibitemNoStop [0]{.\EOS\space}%
\providecommand \EOS [0]{\spacefactor3000\relax}%
\providecommand \BibitemShut  [1]{\csname bibitem#1\endcsname}%
\let\auto@bib@innerbib\@empty
%</preamble>
\bibitem [{\citenamefont {White}(1992)}]{DMRG1}%
  \BibitemOpen
  \bibfield  {author} {\bibinfo {author} {\bibfnamefont {S.~R.}\ \bibnamefont {White}},\ }\href {\doibase 10.1103/PhysRevLett.69.2863} {\bibfield  {journal} {\bibinfo  {journal} {Phys. Rev. Lett.}\ }\textbf {\bibinfo {volume} {69}},\ \bibinfo {pages} {2863} (\bibinfo {year} {1992})}\BibitemShut {NoStop}%
\bibitem [{\citenamefont {\"Ostlund}\ and\ \citenamefont {Rommer}(1995)}]{DMRG2}%
  \BibitemOpen
  \bibfield  {author} {\bibinfo {author} {\bibfnamefont {S.}~\bibnamefont {\"Ostlund}}\ and\ \bibinfo {author} {\bibfnamefont {S.}~\bibnamefont {Rommer}},\ }\href {\doibase 10.1103/PhysRevLett.75.3537} {\bibfield  {journal} {\bibinfo  {journal} {Phys. Rev. Lett.}\ }\textbf {\bibinfo {volume} {75}},\ \bibinfo {pages} {3537} (\bibinfo {year} {1995})}\BibitemShut {NoStop}%
\bibitem [{\citenamefont {Affleck}\ \emph {et~al.}(2004)\citenamefont {Affleck}, \citenamefont {Kennedy}, \citenamefont {Lieb},\ and\ \citenamefont {Tasaki}}]{affleck2004rigorous}%
  \BibitemOpen
  \bibfield  {author} {\bibinfo {author} {\bibfnamefont {I.}~\bibnamefont {Affleck}}, \bibinfo {author} {\bibfnamefont {T.}~\bibnamefont {Kennedy}}, \bibinfo {author} {\bibfnamefont {E.~H.}\ \bibnamefont {Lieb}}, \ and\ \bibinfo {author} {\bibfnamefont {H.}~\bibnamefont {Tasaki}},\ }\href@noop {} {\bibfield  {journal} {\bibinfo  {journal} {Condensed Matter Physics and Exactly Soluble Models: Selecta of Elliott H. Lieb}\ ,\ \bibinfo {pages} {249}} (\bibinfo {year} {2004})}\BibitemShut {NoStop}%
\bibitem [{\citenamefont {Perez-Garcia}\ \emph {et~al.}(2007)\citenamefont {Perez-Garcia}, \citenamefont {Verstraete}, \citenamefont {Wolf},\ and\ \citenamefont {Cirac}}]{mps07}%
  \BibitemOpen
  \bibfield  {author} {\bibinfo {author} {\bibfnamefont {D.}~\bibnamefont {Perez-Garcia}}, \bibinfo {author} {\bibfnamefont {F.}~\bibnamefont {Verstraete}}, \bibinfo {author} {\bibfnamefont {M.}~\bibnamefont {Wolf}}, \ and\ \bibinfo {author} {\bibfnamefont {J.}~\bibnamefont {Cirac}},\ }\href@noop {} {\bibfield  {journal} {\bibinfo  {journal} {Quantum Inf. Comput.}\ ,\ \bibinfo {pages} {407}} (\bibinfo {year} {2007})}\BibitemShut {NoStop}%
\bibitem [{\citenamefont {Hastings}(2007)}]{Hastings_2007}%
  \BibitemOpen
  \bibfield  {author} {\bibinfo {author} {\bibfnamefont {M.~B.}\ \bibnamefont {Hastings}},\ }\href {\doibase 10.1088/1742-5468/2007/08/P08024} {\bibfield  {journal} {\bibinfo  {journal} {J. Stat. Mech.: Theory Exp.}\ }\textbf {\bibinfo {volume} {2007}},\ \bibinfo {pages} {P08024} (\bibinfo {year} {2007})}\BibitemShut {NoStop}%
\bibitem [{\citenamefont {Fannes}\ \emph {et~al.}(1992)\citenamefont {Fannes}, \citenamefont {Nachtergaele},\ and\ \citenamefont {Werner}}]{fannes1992finitely}%
  \BibitemOpen
  \bibfield  {author} {\bibinfo {author} {\bibfnamefont {M.}~\bibnamefont {Fannes}}, \bibinfo {author} {\bibfnamefont {B.}~\bibnamefont {Nachtergaele}}, \ and\ \bibinfo {author} {\bibfnamefont {R.~F.}\ \bibnamefont {Werner}},\ }\href@noop {} {\bibfield  {journal} {\bibinfo  {journal} {Commun. Math. Phys.}\ }\textbf {\bibinfo {volume} {144}},\ \bibinfo {pages} {443} (\bibinfo {year} {1992})}\BibitemShut {NoStop}%
\bibitem [{\citenamefont {Pollmann}\ \emph {et~al.}(2010)\citenamefont {Pollmann}, \citenamefont {Turner}, \citenamefont {Berg},\ and\ \citenamefont {Oshikawa}}]{PhysRevB.81.064439}%
  \BibitemOpen
  \bibfield  {author} {\bibinfo {author} {\bibfnamefont {F.}~\bibnamefont {Pollmann}}, \bibinfo {author} {\bibfnamefont {A.~M.}\ \bibnamefont {Turner}}, \bibinfo {author} {\bibfnamefont {E.}~\bibnamefont {Berg}}, \ and\ \bibinfo {author} {\bibfnamefont {M.}~\bibnamefont {Oshikawa}},\ }\href {\doibase 10.1103/PhysRevB.81.064439} {\bibfield  {journal} {\bibinfo  {journal} {Phys. Rev. B}\ }\textbf {\bibinfo {volume} {81}},\ \bibinfo {pages} {064439} (\bibinfo {year} {2010})}\BibitemShut {NoStop}%
\bibitem [{\citenamefont {Gu}\ \emph {et~al.}(2009)\citenamefont {Gu}, \citenamefont {Levin}, \citenamefont {Swingle},\ and\ \citenamefont {Wen}}]{Pepslevinwen1}%
  \BibitemOpen
  \bibfield  {author} {\bibinfo {author} {\bibfnamefont {Z.-C.}\ \bibnamefont {Gu}}, \bibinfo {author} {\bibfnamefont {M.}~\bibnamefont {Levin}}, \bibinfo {author} {\bibfnamefont {B.}~\bibnamefont {Swingle}}, \ and\ \bibinfo {author} {\bibfnamefont {X.-G.}\ \bibnamefont {Wen}},\ }\href {\doibase 10.1103/PhysRevB.79.085118} {\bibfield  {journal} {\bibinfo  {journal} {Phys. Rev. B}\ }\textbf {\bibinfo {volume} {79}},\ \bibinfo {pages} {085118} (\bibinfo {year} {2009})}\BibitemShut {NoStop}%
\bibitem [{\citenamefont {Soejima}\ \emph {et~al.}(2020)\citenamefont {Soejima}, \citenamefont {Siva}, \citenamefont {Bultinck}, \citenamefont {Chatterjee}, \citenamefont {Pollmann},\ and\ \citenamefont {Zaletel}}]{Pepslevinwen2}%
  \BibitemOpen
  \bibfield  {author} {\bibinfo {author} {\bibfnamefont {T.}~\bibnamefont {Soejima}}, \bibinfo {author} {\bibfnamefont {K.}~\bibnamefont {Siva}}, \bibinfo {author} {\bibfnamefont {N.}~\bibnamefont {Bultinck}}, \bibinfo {author} {\bibfnamefont {S.}~\bibnamefont {Chatterjee}}, \bibinfo {author} {\bibfnamefont {F.}~\bibnamefont {Pollmann}}, \ and\ \bibinfo {author} {\bibfnamefont {M.~P.}\ \bibnamefont {Zaletel}},\ }\href {\doibase 10.1103/PhysRevB.101.085117} {\bibfield  {journal} {\bibinfo  {journal} {Phys. Rev. B}\ }\textbf {\bibinfo {volume} {101}},\ \bibinfo {pages} {085117} (\bibinfo {year} {2020})}\BibitemShut {NoStop}%
\bibitem [{\citenamefont {Schuch}\ \emph {et~al.}(2010)\citenamefont {Schuch}, \citenamefont {Cirac},\ and\ \citenamefont {Pérez-García}}]{SCHUCH20102153}%
  \BibitemOpen
  \bibfield  {author} {\bibinfo {author} {\bibfnamefont {N.}~\bibnamefont {Schuch}}, \bibinfo {author} {\bibfnamefont {I.}~\bibnamefont {Cirac}}, \ and\ \bibinfo {author} {\bibfnamefont {D.}~\bibnamefont {Pérez-García}},\ }\href {\doibase https://doi.org/10.1016/j.aop.2010.05.008} {\bibfield  {journal} {\bibinfo  {journal} {Ann. Phys. (N. Y.)}\ }\textbf {\bibinfo {volume} {325}},\ \bibinfo {pages} {2153} (\bibinfo {year} {2010})}\BibitemShut {NoStop}%
\bibitem [{\citenamefont {Şahinoğlu}\ \emph {et~al.}(2021)\citenamefont {Şahinoğlu}, \citenamefont {Williamson}, \citenamefont {Bultinck}, \citenamefont {Mariën}, \citenamefont {Haegeman}, \citenamefont {Schuch},\ and\ \citenamefont {Verstraete}}]{MPOinjectivePEPS}%
  \BibitemOpen
  \bibfield  {author} {\bibinfo {author} {\bibfnamefont {M.~B.}\ \bibnamefont {Şahinoğlu}}, \bibinfo {author} {\bibfnamefont {D.}~\bibnamefont {Williamson}}, \bibinfo {author} {\bibfnamefont {N.}~\bibnamefont {Bultinck}}, \bibinfo {author} {\bibfnamefont {M.}~\bibnamefont {Mariën}}, \bibinfo {author} {\bibfnamefont {J.}~\bibnamefont {Haegeman}}, \bibinfo {author} {\bibfnamefont {N.}~\bibnamefont {Schuch}}, \ and\ \bibinfo {author} {\bibfnamefont {F.}~\bibnamefont {Verstraete}},\ }\href {\doibase 10.1007/s00023-020-00992-4} {\bibfield  {journal} {\bibinfo  {journal} {Ann. Henri Poincare}\ }\textbf {\bibinfo {volume} {22}},\ \bibinfo {pages} {563–592} (\bibinfo {year} {2021})}\BibitemShut {NoStop}%
\bibitem [{\citenamefont {Hastings}\ and\ \citenamefont {Wen}(2005)}]{LU1}%
  \BibitemOpen
  \bibfield  {author} {\bibinfo {author} {\bibfnamefont {M.~B.}\ \bibnamefont {Hastings}}\ and\ \bibinfo {author} {\bibfnamefont {X.-G.}\ \bibnamefont {Wen}},\ }\href {\doibase 10.1103/PhysRevB.72.045141} {\bibfield  {journal} {\bibinfo  {journal} {Phys. Rev. B}\ }\textbf {\bibinfo {volume} {72}},\ \bibinfo {pages} {045141} (\bibinfo {year} {2005})}\BibitemShut {NoStop}%
\bibitem [{\citenamefont {Chen}\ \emph {et~al.}(2010)\citenamefont {Chen}, \citenamefont {Gu},\ and\ \citenamefont {Wen}}]{LU2}%
  \BibitemOpen
  \bibfield  {author} {\bibinfo {author} {\bibfnamefont {X.}~\bibnamefont {Chen}}, \bibinfo {author} {\bibfnamefont {Z.-C.}\ \bibnamefont {Gu}}, \ and\ \bibinfo {author} {\bibfnamefont {X.-G.}\ \bibnamefont {Wen}},\ }\href {\doibase 10.1103/PhysRevB.82.155138} {\bibfield  {journal} {\bibinfo  {journal} {Phys. Rev. B}\ }\textbf {\bibinfo {volume} {82}},\ \bibinfo {pages} {155138} (\bibinfo {year} {2010})}\BibitemShut {NoStop}%
\bibitem [{\citenamefont {Verstraete}\ \emph {et~al.}(2005)\citenamefont {Verstraete}, \citenamefont {Cirac}, \citenamefont {Latorre}, \citenamefont {Rico},\ and\ \citenamefont {Wolf}}]{PhysRevLett.94.140601}%
  \BibitemOpen
  \bibfield  {author} {\bibinfo {author} {\bibfnamefont {F.}~\bibnamefont {Verstraete}}, \bibinfo {author} {\bibfnamefont {J.~I.}\ \bibnamefont {Cirac}}, \bibinfo {author} {\bibfnamefont {J.~I.}\ \bibnamefont {Latorre}}, \bibinfo {author} {\bibfnamefont {E.}~\bibnamefont {Rico}}, \ and\ \bibinfo {author} {\bibfnamefont {M.~M.}\ \bibnamefont {Wolf}},\ }\href {\doibase 10.1103/PhysRevLett.94.140601} {\bibfield  {journal} {\bibinfo  {journal} {Phys. Rev. Lett.}\ }\textbf {\bibinfo {volume} {94}},\ \bibinfo {pages} {140601} (\bibinfo {year} {2005})}\BibitemShut {NoStop}%
\bibitem [{\citenamefont {Chen}\ \emph {et~al.}(2011)\citenamefont {Chen}, \citenamefont {Gu},\ and\ \citenamefont {Wen}}]{PhysRevB.83.035107}%
  \BibitemOpen
  \bibfield  {author} {\bibinfo {author} {\bibfnamefont {X.}~\bibnamefont {Chen}}, \bibinfo {author} {\bibfnamefont {Z.-C.}\ \bibnamefont {Gu}}, \ and\ \bibinfo {author} {\bibfnamefont {X.-G.}\ \bibnamefont {Wen}},\ }\href {\doibase 10.1103/PhysRevB.83.035107} {\bibfield  {journal} {\bibinfo  {journal} {Phys. Rev. B}\ }\textbf {\bibinfo {volume} {83}},\ \bibinfo {pages} {035107} (\bibinfo {year} {2011})}\BibitemShut {NoStop}%
\bibitem [{\citenamefont {Schuch}\ \emph {et~al.}(2011)\citenamefont {Schuch}, \citenamefont {P\'erez-Garc\'{\i}a},\ and\ \citenamefont {Cirac}}]{MPSphase2}%
  \BibitemOpen
  \bibfield  {author} {\bibinfo {author} {\bibfnamefont {N.}~\bibnamefont {Schuch}}, \bibinfo {author} {\bibfnamefont {D.}~\bibnamefont {P\'erez-Garc\'{\i}a}}, \ and\ \bibinfo {author} {\bibfnamefont {I.}~\bibnamefont {Cirac}},\ }\href {\doibase 10.1103/PhysRevB.84.165139} {\bibfield  {journal} {\bibinfo  {journal} {Phys. Rev. B}\ }\textbf {\bibinfo {volume} {84}},\ \bibinfo {pages} {165139} (\bibinfo {year} {2011})}\BibitemShut {NoStop}%
\bibitem [{\citenamefont {Hastings}(2011)}]{PhysRevLett.107.210501}%
  \BibitemOpen
  \bibfield  {author} {\bibinfo {author} {\bibfnamefont {M.~B.}\ \bibnamefont {Hastings}},\ }\href {\doibase 10.1103/PhysRevLett.107.210501} {\bibfield  {journal} {\bibinfo  {journal} {Phys. Rev. Lett.}\ }\textbf {\bibinfo {volume} {107}},\ \bibinfo {pages} {210501} (\bibinfo {year} {2011})}\BibitemShut {NoStop}%
\bibitem [{\citenamefont {König}\ and\ \citenamefont {Pastawski}(2014)}]{mixedphase0}%
  \BibitemOpen
  \bibfield  {author} {\bibinfo {author} {\bibfnamefont {R.}~\bibnamefont {König}}\ and\ \bibinfo {author} {\bibfnamefont {F.}~\bibnamefont {Pastawski}},\ }\href {\doibase 10.1103/physrevb.90.045101} {\bibfield  {journal} {\bibinfo  {journal} {Phys. Rev. B}\ }\textbf {\bibinfo {volume} {90}} (\bibinfo {year} {2014}),\ 10.1103/physrevb.90.045101}\BibitemShut {NoStop}%
\bibitem [{\citenamefont {Coser}\ and\ \citenamefont {Pérez-García}(2019)}]{mixedphase3}%
  \BibitemOpen
  \bibfield  {author} {\bibinfo {author} {\bibfnamefont {A.}~\bibnamefont {Coser}}\ and\ \bibinfo {author} {\bibfnamefont {D.}~\bibnamefont {Pérez-García}},\ }\href {\doibase 10.22331/q-2019-08-12-174} {\bibfield  {journal} {\bibinfo  {journal} {Quantum}\ }\textbf {\bibinfo {volume} {3}},\ \bibinfo {pages} {174} (\bibinfo {year} {2019})}\BibitemShut {NoStop}%
\bibitem [{\citenamefont {Sang}\ \emph {et~al.}()\citenamefont {Sang}, \citenamefont {Zou},\ and\ \citenamefont {Hsieh}}]{mixedphase1}%
  \BibitemOpen
  \bibfield  {author} {\bibinfo {author} {\bibfnamefont {S.}~\bibnamefont {Sang}}, \bibinfo {author} {\bibfnamefont {Y.}~\bibnamefont {Zou}}, \ and\ \bibinfo {author} {\bibfnamefont {T.~H.}\ \bibnamefont {Hsieh}},\ }\href {https://arxiv.org/abs/2310.08639} {\bibinfo  {journal} {arXiv:2310.08639 [quant-ph]}\ }\BibitemShut {NoStop}%
\bibitem [{\citenamefont {Xue}\ \emph {et~al.}()\citenamefont {Xue}, \citenamefont {Lee},\ and\ \citenamefont {Bao}}]{mixedphase2}%
  \BibitemOpen
\bibfield  {journal} {  }\bibfield  {author} {\bibinfo {author} {\bibfnamefont {H.}~\bibnamefont {Xue}}, \bibinfo {author} {\bibfnamefont {J.~Y.}\ \bibnamefont {Lee}}, \ and\ \bibinfo {author} {\bibfnamefont {Y.}~\bibnamefont {Bao}},\ }\href {https://arxiv.org/abs/2403.17069} {\bibinfo  {journal} {arXiv:2403.17069 [cond-mat.str-el]}\ }\BibitemShut {NoStop}%
\bibitem [{\citenamefont {Ruiz-de Alarcón}\ \emph {et~al.}(2024)\citenamefont {Ruiz-de Alarcón}, \citenamefont {Garre-Rubio}, \citenamefont {Molnár},\ and\ \citenamefont {Pérez-García}}]{MPOalg2}%
  \BibitemOpen
\bibfield  {journal} {  }\bibfield  {author} {\bibinfo {author} {\bibfnamefont {A.}~\bibnamefont {Ruiz-de Alarcón}}, \bibinfo {author} {\bibfnamefont {J.}~\bibnamefont {Garre-Rubio}}, \bibinfo {author} {\bibfnamefont {A.}~\bibnamefont {Molnár}}, \ and\ \bibinfo {author} {\bibfnamefont {D.}~\bibnamefont {Pérez-García}},\ }\href {\doibase 10.1007/s11005-024-01778-z} {\bibfield  {journal} {\bibinfo  {journal} {Lett. Math. Phys.}\ }\textbf {\bibinfo {volume} {114}} (\bibinfo {year} {2024}),\ 10.1007/s11005-024-01778-z}\BibitemShut {NoStop}%
\bibitem [{\citenamefont {Hastings}(2006)}]{MPOGibbs1}%
  \BibitemOpen
  \bibfield  {author} {\bibinfo {author} {\bibfnamefont {M.~B.}\ \bibnamefont {Hastings}},\ }\href {\doibase 10.1103/PhysRevB.73.085115} {\bibfield  {journal} {\bibinfo  {journal} {Phys. Rev. B}\ }\textbf {\bibinfo {volume} {73}},\ \bibinfo {pages} {085115} (\bibinfo {year} {2006})}\BibitemShut {NoStop}%
\bibitem [{\citenamefont {Gondolf}\ \emph {et~al.}()\citenamefont {Gondolf}, \citenamefont {Scalet}, \citenamefont {de~Alarcon}, \citenamefont {Alhambra},\ and\ \citenamefont {Capel}}]{MPOGibbs2}%
  \BibitemOpen
  \bibfield  {author} {\bibinfo {author} {\bibfnamefont {P.}~\bibnamefont {Gondolf}}, \bibinfo {author} {\bibfnamefont {S.~O.}\ \bibnamefont {Scalet}}, \bibinfo {author} {\bibfnamefont {A.~R.}\ \bibnamefont {de~Alarcon}}, \bibinfo {author} {\bibfnamefont {A.~M.}\ \bibnamefont {Alhambra}}, \ and\ \bibinfo {author} {\bibfnamefont {A.}~\bibnamefont {Capel}},\ }\href {https://arxiv.org/abs/2402.18500} {\bibinfo  {journal} {arXiv:2402.18500 [quant-ph]}\ }\BibitemShut {NoStop}%
\bibitem [{\citenamefont {Cui}\ \emph {et~al.}(2015)\citenamefont {Cui}, \citenamefont {Cirac},\ and\ \citenamefont {Bañuls}}]{MPOdissipative}%
  \BibitemOpen
\bibfield  {journal} {  }\bibfield  {author} {\bibinfo {author} {\bibfnamefont {J.}~\bibnamefont {Cui}}, \bibinfo {author} {\bibfnamefont {J.~I.}\ \bibnamefont {Cirac}}, \ and\ \bibinfo {author} {\bibfnamefont {M.~C.}\ \bibnamefont {Bañuls}},\ }\href {\doibase 10.1103/physrevlett.114.220601} {\bibfield  {journal} {\bibinfo  {journal} {Phys. Rev. Lett.}\ }\textbf {\bibinfo {volume} {114}} (\bibinfo {year} {2015}),\ 10.1103/physrevlett.114.220601}\BibitemShut {NoStop}%
\bibitem [{\citenamefont {Chen}\ \emph {et~al.}()\citenamefont {Chen}, \citenamefont {Kato},\ and\ \citenamefont {Brandão}}]{MPDOCMI}%
  \BibitemOpen
  \bibfield  {author} {\bibinfo {author} {\bibfnamefont {C.-F.}\ \bibnamefont {Chen}}, \bibinfo {author} {\bibfnamefont {K.}~\bibnamefont {Kato}}, \ and\ \bibinfo {author} {\bibfnamefont {F.~G. S.~L.}\ \bibnamefont {Brandão}},\ }\href {https://arxiv.org/abs/2010.14682} {\bibinfo  {journal} {arXiv:2010.14682 [quant-ph]}\ }\BibitemShut {NoStop}%
\bibitem [{\citenamefont {Cirac}\ \emph {et~al.}(2011)\citenamefont {Cirac}, \citenamefont {Poilblanc}, \citenamefont {Schuch},\ and\ \citenamefont {Verstraete}}]{PEPSboundary0}%
  \BibitemOpen
\bibfield  {journal} {  }\bibfield  {author} {\bibinfo {author} {\bibfnamefont {J.~I.}\ \bibnamefont {Cirac}}, \bibinfo {author} {\bibfnamefont {D.}~\bibnamefont {Poilblanc}}, \bibinfo {author} {\bibfnamefont {N.}~\bibnamefont {Schuch}}, \ and\ \bibinfo {author} {\bibfnamefont {F.}~\bibnamefont {Verstraete}},\ }\href {\doibase 10.1103/PhysRevB.83.245134} {\bibfield  {journal} {\bibinfo  {journal} {Phys. Rev. B}\ }\textbf {\bibinfo {volume} {83}},\ \bibinfo {pages} {245134} (\bibinfo {year} {2011})}\BibitemShut {NoStop}%
\bibitem [{\citenamefont {Schuch}\ \emph {et~al.}(2013)\citenamefont {Schuch}, \citenamefont {Poilblanc}, \citenamefont {Cirac},\ and\ \citenamefont {P\'erez-Garc\'{\i}a}}]{PEPSboundary}%
  \BibitemOpen
  \bibfield  {author} {\bibinfo {author} {\bibfnamefont {N.}~\bibnamefont {Schuch}}, \bibinfo {author} {\bibfnamefont {D.}~\bibnamefont {Poilblanc}}, \bibinfo {author} {\bibfnamefont {J.~I.}\ \bibnamefont {Cirac}}, \ and\ \bibinfo {author} {\bibfnamefont {D.}~\bibnamefont {P\'erez-Garc\'{\i}a}},\ }\href {\doibase 10.1103/PhysRevLett.111.090501} {\bibfield  {journal} {\bibinfo  {journal} {Phys. Rev. Lett.}\ }\textbf {\bibinfo {volume} {111}},\ \bibinfo {pages} {090501} (\bibinfo {year} {2013})}\BibitemShut {NoStop}%
\bibitem [{\citenamefont {Cirac}\ \emph {et~al.}(2017)\citenamefont {Cirac}, \citenamefont {Pérez-García}, \citenamefont {Schuch},\ and\ \citenamefont {Verstraete}}]{CIRAC2017100}%
  \BibitemOpen
  \bibfield  {author} {\bibinfo {author} {\bibfnamefont {J.}~\bibnamefont {Cirac}}, \bibinfo {author} {\bibfnamefont {D.}~\bibnamefont {Pérez-García}}, \bibinfo {author} {\bibfnamefont {N.}~\bibnamefont {Schuch}}, \ and\ \bibinfo {author} {\bibfnamefont {F.}~\bibnamefont {Verstraete}},\ }\href {\doibase https://doi.org/10.1016/j.aop.2016.12.030} {\bibfield  {journal} {\bibinfo  {journal} {Ann. Phys. (N. Y.)}\ }\textbf {\bibinfo {volume} {378}},\ \bibinfo {pages} {100} (\bibinfo {year} {2017})}\BibitemShut {NoStop}%
\bibitem [{\citenamefont {Koashi}\ and\ \citenamefont {Imoto}(2002)}]{KIdecomposition}%
  \BibitemOpen
  \bibfield  {author} {\bibinfo {author} {\bibfnamefont {M.}~\bibnamefont {Koashi}}\ and\ \bibinfo {author} {\bibfnamefont {N.}~\bibnamefont {Imoto}},\ }\href@noop {} {\bibfield  {journal} {\bibinfo  {journal} {Phys. Rev. A}\ }\textbf {\bibinfo {volume} {66}},\ \bibinfo {pages} {022318} (\bibinfo {year} {2002})}\BibitemShut {NoStop}%
\bibitem [{\citenamefont {Hayden}\ \emph {et~al.}(2004)\citenamefont {Hayden}, \citenamefont {Jozsa}, \citenamefont {Petz},\ and\ \citenamefont {Winter}}]{Hayden2004-os}%
  \BibitemOpen
  \bibfield  {author} {\bibinfo {author} {\bibfnamefont {P.}~\bibnamefont {Hayden}}, \bibinfo {author} {\bibfnamefont {R.}~\bibnamefont {Jozsa}}, \bibinfo {author} {\bibfnamefont {D.}~\bibnamefont {Petz}}, \ and\ \bibinfo {author} {\bibfnamefont {A.}~\bibnamefont {Winter}},\ }\href@noop {} {\bibfield  {journal} {\bibinfo  {journal} {Commun. Math. Phys.}\ }\textbf {\bibinfo {volume} {246}},\ \bibinfo {pages} {359} (\bibinfo {year} {2004})}\BibitemShut {NoStop}%
\bibitem [{\citenamefont {Kato}()}]{kato2023}%
  \BibitemOpen
  \bibfield  {author} {\bibinfo {author} {\bibfnamefont {K.}~\bibnamefont {Kato}},\ }\href {https://arxiv.org/abs/2309.07434} {\bibinfo  {journal} {arXiv:2309.07434 [quant-ph]}\ }\BibitemShut {NoStop}%
\bibitem [{\citenamefont {Fisher}(1922)}]{Fisher1}%
  \BibitemOpen
\bibfield  {journal} {  }\bibfield  {author} {\bibinfo {author} {\bibfnamefont {R.~A.}\ \bibnamefont {Fisher}},\ }\href@noop {} {\bibfield  {journal} {\bibinfo  {journal} {Philos. Trans. R. Soc. Lond. A}\ }\textbf {\bibinfo {volume} {222}},\ \bibinfo {pages} {309} (\bibinfo {year} {1922})}\BibitemShut {NoStop}%
\bibitem [{\citenamefont {Cover}\ and\ \citenamefont {Thomas}(2006)}]{Cover2006}%
  \BibitemOpen
  \bibfield  {author} {\bibinfo {author} {\bibfnamefont {T.~M.}\ \bibnamefont {Cover}}\ and\ \bibinfo {author} {\bibfnamefont {J.~A.}\ \bibnamefont {Thomas}},\ }\href@noop {} {\emph {\bibinfo {title} {Elements of Information Theory 2nd Edition (Wiley Series in Telecommunications and Signal Processing)}}}\ (\bibinfo  {publisher} {Wiley-Interscience},\ \bibinfo {year} {2006})\BibitemShut {NoStop}%
\bibitem [{\citenamefont {Petz}(1986)}]{petz1986sufficient}%
  \BibitemOpen
  \bibfield  {author} {\bibinfo {author} {\bibfnamefont {D.}~\bibnamefont {Petz}},\ }\href@noop {} {\bibfield  {journal} {\bibinfo  {journal} {Commun. Math. Phys.}\ }\textbf {\bibinfo {volume} {105}},\ \bibinfo {pages} {123} (\bibinfo {year} {1986})}\BibitemShut {NoStop}%
\bibitem [{\citenamefont {Petz}(1988)}]{10.1093/qmath/39.1.97}%
  \BibitemOpen
  \bibfield  {author} {\bibinfo {author} {\bibfnamefont {D.}~\bibnamefont {Petz}},\ }\href {\doibase 10.1093/qmath/39.1.97} {\bibfield  {journal} {\bibinfo  {journal} {Q. J. Math.}\ }\textbf {\bibinfo {volume} {39}},\ \bibinfo {pages} {97} (\bibinfo {year} {1988})}\BibitemShut {NoStop}%
\bibitem [{\citenamefont {Molnar}\ \emph {et~al.}()\citenamefont {Molnar}, \citenamefont {de~Alarcón}, \citenamefont {Garre-Rubio}, \citenamefont {Schuch}, \citenamefont {Cirac},\ and\ \citenamefont {Pérez-García}}]{MPOalg1}%
  \BibitemOpen
  \bibfield  {author} {\bibinfo {author} {\bibfnamefont {A.}~\bibnamefont {Molnar}}, \bibinfo {author} {\bibfnamefont {A.~R.}\ \bibnamefont {de~Alarcón}}, \bibinfo {author} {\bibfnamefont {J.}~\bibnamefont {Garre-Rubio}}, \bibinfo {author} {\bibfnamefont {N.}~\bibnamefont {Schuch}}, \bibinfo {author} {\bibfnamefont {J.~I.}\ \bibnamefont {Cirac}}, \ and\ \bibinfo {author} {\bibfnamefont {D.}~\bibnamefont {Pérez-García}},\ }\href {https://arxiv.org/abs/2204.05940} {\bibinfo  {journal} {arXiv:2204.05940 [quant-ph]}\ }\BibitemShut {NoStop}%
\bibitem [{\citenamefont {Schumacher}(1995)}]{PhysRevA.51.2738}%
  \BibitemOpen
\bibfield  {journal} {  }\bibfield  {author} {\bibinfo {author} {\bibfnamefont {B.}~\bibnamefont {Schumacher}},\ }\href {\doibase 10.1103/PhysRevA.51.2738} {\bibfield  {journal} {\bibinfo  {journal} {Phys. Rev. A}\ }\textbf {\bibinfo {volume} {51}},\ \bibinfo {pages} {2738} (\bibinfo {year} {1995})}\BibitemShut {NoStop}%
\bibitem [{Note1()}]{Note1}%
  \BibitemOpen
  \bibinfo {note} {In the statement in Theorem 4.15. in Ref.~\cite {CIRAC2017100}, the state is written as $\DOTSB \tsum \slimits@ _{i=1}^d\lambda _iP_i^{(L)}e^{-H_L}$ where $H_N$ is a nearest-neighbor commuting Hamiltonian. However, the explicit construction in the proof tells us that $e^{-H_L}$ is given as $\mu ^{\otimes L}$ for $\mu >0$ (in this paper we use $\Omega $ instead of $\mu $), which is not interacting~\cite {private}.}\BibitemShut {Stop}%
\bibitem [{\citenamefont {Mosonyi}\ and\ \citenamefont {Petz}(2004)}]{Mosonyi2004-wo}%
  \BibitemOpen
  \bibfield  {author} {\bibinfo {author} {\bibfnamefont {M.}~\bibnamefont {Mosonyi}}\ and\ \bibinfo {author} {\bibfnamefont {D.}~\bibnamefont {Petz}},\ }\href@noop {} {\bibfield  {journal} {\bibinfo  {journal} {Lett. Math. Phys.}\ }\textbf {\bibinfo {volume} {68}},\ \bibinfo {pages} {19} (\bibinfo {year} {2004})}\BibitemShut {NoStop}%
\bibitem [{\citenamefont {Jenčová}\ and\ \citenamefont {Petz}(2006)}]{Jencova2006-on}%
  \BibitemOpen
  \bibfield  {author} {\bibinfo {author} {\bibfnamefont {A.}~\bibnamefont {Jenčová}}\ and\ \bibinfo {author} {\bibfnamefont {D.}~\bibnamefont {Petz}},\ }\href@noop {} {\bibfield  {journal} {\bibinfo  {journal} {Commun. Math. Phys.}\ }\textbf {\bibinfo {volume} {263}},\ \bibinfo {pages} {259} (\bibinfo {year} {2006})}\BibitemShut {NoStop}%
\bibitem [{Note2()}]{Note2}%
  \BibitemOpen
  \bibinfo {note} {When we speak about the minimal sufficient subalgebra for a partition $AB~\Lambda _L$, we implicitly assume that ${\protect \mathcal H}_B$ is restricted to ${\protect \rm supp}(\rho _B)$ so that $\rho _B>0$.}\BibitemShut {Stop}%
\bibitem [{\citenamefont {Leifer}\ and\ \citenamefont {Poulin}(2008)}]{leifer2008quantum}%
  \BibitemOpen
  \bibfield  {author} {\bibinfo {author} {\bibfnamefont {M.~S.}\ \bibnamefont {Leifer}}\ and\ \bibinfo {author} {\bibfnamefont {D.}~\bibnamefont {Poulin}},\ }\href@noop {} {\bibfield  {journal} {\bibinfo  {journal} {Ann. Phys. (N. Y.)}\ }\textbf {\bibinfo {volume} {323}},\ \bibinfo {pages} {1899} (\bibinfo {year} {2008})}\BibitemShut {NoStop}%
\bibitem [{\citenamefont {Brown}\ and\ \citenamefont {Poulin}()}]{brown2012}%
  \BibitemOpen
  \bibfield  {author} {\bibinfo {author} {\bibfnamefont {W.}~\bibnamefont {Brown}}\ and\ \bibinfo {author} {\bibfnamefont {D.}~\bibnamefont {Poulin}},\ }\href {https://arxiv.org/abs/1206.0755} {\bibinfo  {journal} {arXiv:1206.0755 [quant-ph]}\ }\BibitemShut {NoStop}%
\bibitem [{\citenamefont {Etingof}\ \emph {et~al.}()\citenamefont {Etingof}, \citenamefont {Nikshych},\ and\ \citenamefont {Ostrik}}]{etingof2017fusioncategories}%
  \BibitemOpen
\bibfield  {journal} {  }\bibfield  {author} {\bibinfo {author} {\bibfnamefont {P.}~\bibnamefont {Etingof}}, \bibinfo {author} {\bibfnamefont {D.}~\bibnamefont {Nikshych}}, \ and\ \bibinfo {author} {\bibfnamefont {V.}~\bibnamefont {Ostrik}},\ }\href {https://arxiv.org/abs/math/0203060} {\bibinfo  {journal} {arXiv:math/0203060 [math.QA]}\ }\BibitemShut {NoStop}%
\bibitem [{\citenamefont {Perez-Garcia}(2024)}]{private}%
  \BibitemOpen
\bibfield  {journal} {  }\bibfield  {author} {\bibinfo {author} {\bibfnamefont {D.}~\bibnamefont {Perez-Garcia}},\ }\href@noop {} {\bibfield  {journal} {\bibinfo  {journal} {{private communication}}\ } (\bibinfo {year} {2024})}\BibitemShut {NoStop}%
\bibitem [{\citenamefont {Kastoryano}\ \emph {et~al.}(2019)\citenamefont {Kastoryano}, \citenamefont {Lucia},\ and\ \citenamefont {Perez-Garcia}}]{Kastoryano_2019}%
  \BibitemOpen
  \bibfield  {author} {\bibinfo {author} {\bibfnamefont {M.~J.}\ \bibnamefont {Kastoryano}}, \bibinfo {author} {\bibfnamefont {A.}~\bibnamefont {Lucia}}, \ and\ \bibinfo {author} {\bibfnamefont {D.}~\bibnamefont {Perez-Garcia}},\ }\href {\doibase 10.1007/s00220-019-03404-9} {\bibfield  {journal} {\bibinfo  {journal} {Commun. Math. Phys.}\ }\textbf {\bibinfo {volume} {366}},\ \bibinfo {pages} {895–926} (\bibinfo {year} {2019})}\BibitemShut {NoStop}%
\bibitem [{\citenamefont {Levin}\ and\ \citenamefont {Nave}(2007)}]{TRG1}%
  \BibitemOpen
  \bibfield  {author} {\bibinfo {author} {\bibfnamefont {M.}~\bibnamefont {Levin}}\ and\ \bibinfo {author} {\bibfnamefont {C.~P.}\ \bibnamefont {Nave}},\ }\href {\doibase 10.1103/physrevlett.99.120601} {\bibfield  {journal} {\bibinfo  {journal} {Phys. Rev. Lett.}\ }\textbf {\bibinfo {volume} {99}} (\bibinfo {year} {2007}),\ 10.1103/physrevlett.99.120601}\BibitemShut {NoStop}%
\bibitem [{\citenamefont {Gu}\ \emph {et~al.}(2008)\citenamefont {Gu}, \citenamefont {Levin},\ and\ \citenamefont {Wen}}]{TRG2}%
  \BibitemOpen
  \bibfield  {author} {\bibinfo {author} {\bibfnamefont {Z.-C.}\ \bibnamefont {Gu}}, \bibinfo {author} {\bibfnamefont {M.}~\bibnamefont {Levin}}, \ and\ \bibinfo {author} {\bibfnamefont {X.-G.}\ \bibnamefont {Wen}},\ }\href {\doibase 10.1103/physrevb.78.205116} {\bibfield  {journal} {\bibinfo  {journal} {Phys. Rev. B}\ }\textbf {\bibinfo {volume} {78}} (\bibinfo {year} {2008}),\ 10.1103/physrevb.78.205116}\BibitemShut {NoStop}%
\end{thebibliography}%

\appendix
\section{Proofs of Lemmas}
 \setcounter{lem}{0} 
\subsection{Proof of Lemma~\ref{lem:Mldefinition}}
    \begin{lem}
    Let $AB$ be a bipartition of $\Lambda_L$ such that $B$ is connected and $|B|=l$. Then, $\cM_\cS^{\min}$ of an MPDO $\rho_{AB}^{(L)}$ on $B$ is the same for all $L\geq l+1$.
    \end{lem}
\begin{proof}
The MPDO is given as
     \begin{equation*}
            \rho^{(L)}=\sum_k\tr\bk{N_k^L}\sum_{\bf i,j} \tr\bk{M^{i_1j_1}_k...M_k^{i_Lj_L}}\ket{i_1...i_L}\bra{j_1...j_L}\,.
        \end{equation*}
    Denote the set~\eqref{eq:cS} for $\rho^{(L)}$, defined on the subsystem $B$, by $\cS$ and that of $\rho^{(l+1)}$ by $\cS'$.     
    We first show that the linear span of $\cS$  is independent of $L\geq l+1$ and matches the linear span of $\cS'$. By definition, each element $\mu_B\in\cS$ takes the form  
     \begin{align*}
            \mu_B&\propto\sum_k\tr\bk{N_k^L}\sum_{\bf i,j} \tr\bk{M^{i_1j_1}_k...M_k^{i_Lj_L}}\ket{i_1...i_l}\bra{j_1...j_l}\,\\
            &\quad\times\bra{j_{l+1}...j_L}O_A\ket{i_{l+1}...i_L}.
        \end{align*}
        Let $X=\bigoplus_kX_k$ be a boundary matrix defined by
            \begin{equation*}
            X_k:=\tr\bk{N_k^L}\sum_{\bf i,j}\bra{j_{l+1}...j_L}O_A\ket{i_{l+1}...i_L}M_k^{i_{l+1}j_{l+1}}...M_k^{i_Lj_L}.
        \end{equation*}
        Then 
         \begin{align*}
            \mu_B&\propto\sum_k\sum_{\bf i,j} \tr\bk{X_kM^{i_1j_1}_k...M_k^{i_lj_l}}\ket{i_1...i_l}\bra{j_1...j_l}\,.
        \end{align*}
        The biCF condition for $M_k^{ij}$ implies that there exists a matrix on $l+1$th site $C_{l+1}(X):=(c_{ji}(X))_{ij}$ such that
        \begin{align*}
           \tr_{l+1}\bk{C_{l+1}(X)\rho^{(l+1)}}\propto \mu_B\,,
        \end{align*}
        and therefore $\cS\subset{\rm span}\cS'$ ($C_{l+1}(X)$ is not necessarily be positive semi-definite). The same arguments hold for $\cS'\subset{\rm span}\cS$, concluding ${\rm span}\cS={\rm span}\cS'$.

        Any state in $\cS$ has the Koashi-Imoto decomposition with respect to the minimal sufficient subalgebra $\cM_{\cS}^{\min}$. However, since it is also an element of ${\rm span}\cS'$, it can also be decomposed according to $\cM_{\cS'}^{\min}$. Due to the minimality of the Koashi-Imoto decomposition, these two decompositions must coincide (up to a unitary transformation). Since the minimal sufficient subalgebras are generated by operators such as $\mu_B^{it}\rho_B^{-it}$, these subalgebras must also be equivalent.   
\end{proof}

\subsection{Proof of Lemma~\ref{vCF=KI}}
    \begin{lem}
    Let $AB$ be a bipartition of $\Lambda_L$ such that $B$ is connected and $|B|=l$. Consider the v-CF of the matrix $W_k^{\alpha_k\beta_k}[l]$
            \begin{align}\label{apeq:vcf1}
           W_k^{\alpha_k\beta_k}[l]\cong\bigoplus_{a'=1}^{r'} W_{k,a'}^{\alpha_k\beta_k}[l] \ot \Omega_{a'}^{(l)}\,
        \end{align}
        with the decomposition of the Hilbert space 
        \begin{equation*}
            \cH_B\cong\bigoplus_{a'=1}^{r'}\cH_{B_{a'}^L}\otimes \cH_{B_{a'}^R}\,.
        \end{equation*}
        Then we have
        \begin{equation}\label{apeq:kidec}
                        \cM^{\min}_l\cong \bigoplus_{a'=1}^{r'}\cB\bk{\cH_{B_{a'}^L}}\ot I_{\cH_{B_{a'}^R}}\,.
        \end{equation}
        Conversely, in the basis such that Eq.~\eqref{apeq:kidec} holds, $W_k^{\alpha_k\beta_k}[l]$ is in the v-CF~\eqref{apeq:kidec}.
    \end{lem}
Before proceeding with the proof, we first establish
\begin{equation}
    {\rm span}\{\mu_B\}={\rm span}\left\{W_k^{\alpha_k\beta_k}[l]\right\}\,.\label{lem:mu=W}
\end{equation}
It can be written as 
\begin{equation*}
\rho_{AB}=\sum_k\tr\bk{N_k^L}    W_k^{\beta_k\alpha_k}[L-l]\ot W_k^{\alpha_k\beta_k}[l]\,.
\end{equation*}
${\rm span}\{\mu_B\}\subset{\rm span}\left\{W_k^{\alpha_k\beta_k}[l]\right\}$ holds because \begin{equation}
        \mu_{B}\propto\sum_k\tr\bk{N_k^{L}}\sum_{a,b}\tr\bk{O_AW_k^{\beta_k\alpha_k}[L-l]}W_k^{\alpha_k\beta_k}[l],
\end{equation}
where $O_A$ is the operator in Eq.~\eqref{eq:cS}. 
We now show the converse for $L=l+1$ (the generalization is straightforward). The MPDO can be rewritten as
\begin{align*}
\rho_{AB}&=\sum_k\tr\bk{N_k^L}    \sum_{i_{l+1},j_{l+1}}\bra{\beta_k}M^{i_{l+1}j_{l+1}}_k\ket{\alpha_k}\\
&\times\ket{i_{l+1}}\bra{j_{l+1}}\ot W_k^{\alpha_k\beta_k}[l]\,.
\end{align*}
Then, by the biCF condition, there exists a matrix $C_{l+1}(|\beta_k\rangle\langle\alpha_k|)$ such that
\begin{equation*}
    \tr_{l+1}\bk{C_{l+1}(|\beta_k\rangle\langle\alpha_k|)\rho_{AB}}\propto W_k^{\alpha_k\beta_k}[l],
\end{equation*}
which implies ${\rm span}\{\mu_B\}\supset{\rm span}\left\{W_k^{\alpha_k\beta_k}[l]\right\}$.  
\begin{proof}
Let $\cH_B\cong\bigoplus_x\cH_{B_x^L}\ot\cH_{B_x^R}$ be the decomposition induced by the v-CF of $W_k^{\alpha_k\beta_k}[l]$. Eq.~\eqref{lem:mu=W} show that, based on the v-CF of $W_k^{\alpha_k\beta_k}[l]$, any $\mu_B\in \cS$ can be decomposed as
\begin{equation}
    \mu_B=\bigoplus_x{\hat \mu}_{B_x^L}\ot\Omega_{B_x^R}\,,
\end{equation} 
where ${\hat \mu}_B\geq0$ acts on $\cH_{B_x^L}$ and $\Omega_{B_x^R}>0$ acts on $\cH_{B_x^R}$. In particular,
\begin{equation}
    \rho_B=\bigoplus_x{\hat \rho}_{B_x^L}\ot\Omega_{B_x^R}\,,
\end{equation} 
with ${\hat \rho}_{B_x^L}>0$. Therefore,
\begin{equation}
    \mu_B^{it}\rho_B^{-it}=\bigoplus_x{\hat \mu}_{B_x^L}^{it}{\hat \rho}_{B_x^L}^{-it}\ot I_{B_x^R}\,,
\end{equation} 
and
\begin{equation}
\cM_B^{\min}=\alg\left\{\mu^{it}_B\rho^{-it}_B, \mu\in\cS, t\in\mathbb{R} \right\} \subset \bigoplus_x\cB\left(\cH_{B_x^L}\right)\ot I_{\cH_{B_x^R}}\,.
\end{equation}
Now, suppose that the inclusion is strict. Then, there exists at least one $x$ such that $\cH_{B_x^L}=\bigoplus_{v_x}\cH_{B_{v_x}^L}$ and
\begin{equation}
  {\hat \mu}_{B_x^L}=\bigoplus_{v_x}  {\hat \mu}_{B_{v_x}^L}
\end{equation}
for all $\mu_B\in\cS$. However, Eq.~\eqref{lem:mu=W} implies that 
\begin{equation}
    W_{k,x}^{\alpha_k\beta_k}[l]=\bigoplus_{v_x}W_{k,v_x}^{\alpha_k\beta_k}[l]\,,
\end{equation}
which contradicts the fact that the CP-map $\cF_x$ from the definition of BNT does not have nontrivial invariant subspaces. Therefore,
\begin{equation}
\cM_B^{\min}= \bigoplus_x\cB\left(\cH_{B_x^L}\right)\ot I_{\cH_{B_x^R}}\,.
\end{equation}
\end{proof}

\subsection{Proof of Lemma~\ref{lem:minalginclusion}}
    \begin{lem}
    Let $l_1+l_2=l$ and $K_{l_1,l_2}$ be the kernel of $\rho_{l_1+l_2}$ in $\supp(\rho_{l_1}\ot\rho_{l_2})$. Then the following holds:
    \begin{equation}
        \cM^{\min}_{l_1+l_2}\oplus0_{K_{l_1,l_2}}\subset\cM_{l_1}^{\min}\ot\cM_{l_2}^{\min}\,.
    \end{equation}
    \end{lem}
    This lemma follows directly from the next lemma.
\begin{lem}\label{lem:mSmS}
Suppose that $S'$ is dominated by $\rho_1\geq0$, and $S$ is dominated (i.e., $\forall\mu\in\cS, \supp(\mu)\subset\supp(\rho_2)$) by $\rho_2>0$. Denote $K:={\rm Ker}(\rho_1)\cap\supp(\rho_2)$. If $\cS'\subset\cS$, 
\begin{equation}
\cM_{\cS'}^{\min}\oplus0_K\subset\cM_\cS^{\min}\,.    
\end{equation}
\end{lem}
\begin{proof}
Let
\begin{equation}
    \cM_S^{\min}=\bigoplus_x\cB(\cH^L_x)\ot I_{\cH^R_x}\,.
\end{equation}
Any element $\mu\in \cS$ can be decomposed as
\begin{equation}
    \mu=\bigoplus_x{\hat \mu}_{\cH^L_x}\ot\omega_{\cH^R_x}\,,
\end{equation}
where $\omega_{\cH^R_x}>0$. Hence we have
\begin{equation}
    \supp(\rho_1)=\bigoplus_x \cH'^L_x\ot\cH^R_x\,,
\end{equation}
where $\cH'^L_x\subset \cH^L_x$. 
This implies one can define
\begin{equation}
    \rho_1^{-it}=\bigoplus_x{\hat \rho}_{\cH'^L_x}^{-it}\ot\omega^{-it}_{\cH^R_x}\,,
\end{equation}
on $\supp(\rho_1)$ and therefore 
\begin{equation}
    \mu'^{it}\rho_1^{-it}|_{\supp(\rho_1)}=\bigoplus_x{\hat {\mu'}}^{it}{\hat \rho}^{-it}_{\cH'^L_x}\ot I_{\cH^R_x}
\end{equation}
for any $\mu'\in S'$. On $\supp(\rho_2)\supset\supp(\rho_1)$, we have $\mu'^{it}\rho_1^{-it}|_{\supp(\rho_1)}\oplus0_K\in\cM^S$ and
\begin{equation}
    \alg\{\mu'^{it}\rho_1^{-it}|_{\supp(\rho_1)}\oplus0_K\}=\cM^{\min}_{\cS'}\oplus0_K\subset\cM_\cS^{\min}\,.
\end{equation}

\end{proof}

\begin{proof}
From Eq.~\eqref{lem:mu=W}, we have
    \begin{align*}
       \spann\{\mu_{2l}\}&=\spann\{W_k^{\alpha_k\beta_k}[2l]\}\\
       &\subset \spann\{W_k^{\alpha_k\beta_k}[l]\ot W_k^{\gamma_k\delta_k}[l] \}=\spann\{\mu_{l}\ot\mu'_l\}\,.
    \end{align*}
The minimal sufficient subalgebra for $\{\mu_l\ot \mu'_l\}$ is $\cM^S_{l}\ot\cM^S_{l}$, 
and thus, by Lemma~\ref{lem:mSmS}, we have 
\begin{equation}
    \cM^S_{2l}\oplus0_K\subset \cM^S_{l}\ot\cM^S_{l}\,.
\end{equation}
\end{proof}

\section{Proof of theorems}
 \setcounter{prop}{0} 
\subsection{Proof of Proposition~\ref{prop:comultiplication}}
\begin{prop}$\Delta:\cA\mapsto \cA\ot\cA$ is a multiplicative co-multiplication on $\cA$, i.e., it satisfies 
    \begin{equation*}
        \Delta(xy)=\Delta(x)\Delta(y) \quad \forall x,y\in\cA\,
    \end{equation*}
    and
\begin{equation*}
    \bk{\id\ot\Delta}\circ\Delta=\bk{\Delta\ot \id}\circ\Delta\,=:\Delta^2\,.
\end{equation*}
\end{prop}
\begin{proof}
By definition, $\Delta=(\pi^{-1}_{l_1}\otimes\pi^{-1}_{l_1})\circ\iota_{l_1+l_2}\circ\pi_{l_1+l_2}$, which is a composition of $*$-homomorphisms and an inclusion homomorphism, is multiplicative. 

Let $l_0$ be some number. Consider two inclusion relations:
\begin{align*}
    \pi_{3l_0}(\cA)&\subset\pi_{l_0}(\cA)\ot\pi_{2l_0}(\cA)\subset\pi_{l_0}(\cA)\ot\left(\pi_{l_0}(\cA)\ot\pi_{l_0}(\cA)\right)\,,\\
    \pi_{3l_0}(\cA)&\subset\pi_{2l_0}(\cA)\ot\pi_{l_0}(\cA)\subset\left(\pi_{l_0}(\cA)\ot\pi_{l_0}(\cA)\right)\ot\pi_{l_0}(\cA)\,.
\end{align*}
These relations imply the following commutative diagram:
\begin{equation*}
\xymatrix{
&\pi_{2l_0}(\cA)\ot\pi_{l_0}(\cA)\ar[rd]^-{\iota_{l_0+l_0}\ot\id}\ar@{}[dd]|{\circlearrowright}&\\
\pi_{3l_0}(\cA)\ar[ur]_-{\iota_{(2l_0)+l_0}}\ar[dr]_-{\iota_{l_0+(2l_0)}}&&\pi_{l_0}(\cA)^{\otimes 3}\\
&\pi_{l_0}(\cA)\ot\pi_{2l_0}(\cA)\ar[ur]_-{\id\ot\iota_{l_0+l_0}}&
}
\end{equation*}
which yields 
\begin{equation*}
\left(\Delta_{l_0+l_0}\ot\id\right)\circ\Delta_{2l_0+l_0}=\left(\id\ot\Delta_{l_0+l_0}\right)\circ\Delta_{l_0+2l_0}\,.
\end{equation*}
The proposition follows from the assumption that  $\Delta_{l+l'}=\Delta$ for any $l,l'$.
\end{proof}

\section{Proof of Theorem~\ref{thm:prebialgbra}}\label{app:prebialgebra}
We have already established that $\cA$ has a comultiplication $\Delta$, so the dual space $\cA^*$ inherits a multiplication defined as 
\begin{equation}
    \delta_x\delta_y(A):=(\delta_x\ot\delta_y)\circ\Delta(A)\quad \forall A\in\cA\,.
\end{equation}
It is known that $\cA$ has a counit if and only if $\cA^*$ has a unit for this multiplication~\cite{MPOalg1}.  
Now, consider $\cA^*_+=\{(a,z)|a\in\cA,z\in\C\}$, the unitization of $\cA^*$ with a multiplication
\begin{equation*}
    (a,z_1)(b,z_2):=(ab+z_2a+z_1b,z_1z_2)\,.
\end{equation*} 
 $\cA_+^*$ is a finite-dimensional unital $C^*$-algebra, and is therefore isomorphic to a matrix subalgebra on a finite-dimensional Hilbert space.  

Denote the injective representation induced by this isomorphism by $\psi$, and consider a tensor
\begin{equation*}
    O:=\sum_{x\in\cB_\cA}\pi(x)\ot\psi(\delta_x)\,,
\end{equation*}
where $\cB_\cA=\{x\}$ is a basis of $\cA$ and $\{\delta_x\}$ is the dual basis of $\cA^*$. 
By concatenating the tensors, we obtain
\begin{align*}
\sum_{xy}\pi(x)\ot\pi(y)\ot&\psi\bk{\delta_x}\psi\bk{\delta_y}\\
&=\bk{\pi^{\ot 2}\ot\psi}\bk{\sum_{xy}(x\ot y)\ot\delta_x\delta_y}\\
    &=\sum_z\pi^{\ot2}\circ\Delta(z)\ot\psi\bk{\delta_z}\,
\end{align*}
and thus the tensor generates MPOs:
\begin{equation*}
    \underbrace{OO...O}_L=\sum_{x\in\cB_\cA}\pi^{\ot L}\circ\Delta^{L-1}(x)\ot\psi(\delta_x)\,.
\end{equation*}

The standard form of the tensor $O$ is
\begin{equation*}
    O=\sum_{i,j=1}^d\ket{i}\bra{j}\ot O^{ij}\,,
\end{equation*}
and comparison to the definition yields 
\begin{equation*}
    O^{ij}=\psi\bk{f^{ij}}\,,\quad f^{ij}:=\sum_x\bra{i}\pi(x)\ket{j}\delta_x\,.
\end{equation*}
The h-CF of $O$ implies the existence of $X>0$ such that
\begin{equation*}
    XO^{ij}X^{-1}=\bigoplus_s O^{ij}_s\ot K_s\,.
\end{equation*}
We redefine $\psi(\cdot):=X\psi(\cdot)X^{-1}$.

For any element $Y=\sum_{ij}y_{ij}\ket{i}\bra{j}\in\cB(\C^d)$,
\begin{equation*}
   \sum_{i,j}{\bar y}_{ij}f^{ij}= \sum_{i,j}{\bar y}_{ij}\sum_x\bra{i}\pi(x)\ket{j}\delta_x=\lr{\langle Y,\pi(x)}\rangle_{HS}\cdot\delta_x\,,
\end{equation*}
where $\langle A,B\rangle_{HS}:=\tr\bk{A^\dagger B}$. 
Since $\pi$ is an injective representation, for any $y\in\cB_\cA$, there is $Y\in\pi(\cA)\subset\cB(\C^d)$ such that
\begin{equation*}
    \lr{\langle Y,\pi(x)}\rangle_{HS}=\delta_y(x)\,.
\end{equation*}
Therefore, it holds that
\begin{align}
    {\rm span}\lr{\{f^{ij}\:\middle|\: i,j=1,...,d }\}&={\rm span}\lr{\{\delta_x\:|\:\delta_x\in\cB_{\cA^*}}\}\nonumber\\
    &=\cA^*\,.\label{eq:spanAstar}
\end{align}
Eq.~\eqref{eq:spanAstar} implies for any $i,j,kl,$ there is $\alpha_{nm}^{ijkl}$ such that  $f^{ij}f^{kl}=\sum_{nm}\alpha^{ijkl}_{nm}f^{nm}$, and thus 
\begin{align*}
    \psi\bk{f^{ij}f^{kl}}&=\bigoplus_s\bk{\sum_{nm}\alpha^{ijkl}_{nm} O^{nm}_s}\ot K_s\,,\\
    \psi\bk{f^{ij}}\psi\bk{f^{kl}}&=\bigoplus_s O^{ij}_sO^{kl}_s\ot K_s^2\,.
\end{align*}
Thus, $\psi\bk{f^{ij}f^{kl}}=\psi\bk{f^{ij}}\psi\bk{f^{kl}}$ only holds when $K_s^2=K_s$. 
Since the subspaces can always be restricted so that $K_s>0$ without altering the MPOs, $K_s$ act as the identity on the subspace denoted by $I_s$.  Note that due to this restriction, $\psi(\cA)$ may be supported on a smaller subspace than $\psi(\cA^*_+)$. We denote the support of $\psi(\cA)$ by $\C^{\tilde D}$ for some ${\tilde D}\in\mathbb{N}$. 

Finally, $\{O_s^{ij}\}$ is a BNT, and there exists a finite length for which the biCF condition~\eqref{eq:biCF} is satisfied. This implies that
\begin{equation*}
   I_{\tilde D} \in \psi(\cA)\,,
\end{equation*}
and thus $\epsilon_{\cA^*}:=\psi^{-1}(I_{\tilde D})$ serves as the unit of $\cA^*$, which is the counit of $\cA$.

Hence, $\cA$ has a multiplicative comultiplication $\Delta$ and a counit, confirming that $\cA$ is a pre-bialgebra.

\subsection{Proof of Theorem~\ref{thm:MPDOstructure}}\label{app:structurethm}
%        One may notice the similarity between the Koashi-Imoto decomposition~\eqref{eq:KI-dec} and the v-CF~\eqref{eq:v-CF}. 
    \begin{lem}\label{lem:KI=vCF}
    The Koashi-Imoto (KI)-decomposition is equivalent to the decomposition in the vCF.
    \end{lem}
This lemma, together with Eq.~\eqref{eq:min=piA}, implies that for any $l$, the vCF can be  rewritten as
\begin{equation*}
    U_lW_k^{\alpha_k\beta_k}[l]U_l^\dagger=\pi_l\bk{{\hat w}_k^{\alpha_k\beta_k}[l]}\Omega^{(l)}\,,\quad {\hat w}_k^{\alpha_k\beta_k}[l]\in\cA\,,
\end{equation*}
where
\begin{align*}
    \pi_l\bk{{\hat w}_k^{\alpha_k\beta_k}[l]}&:=\bigoplus_{a=1}^{r}W_{k,a}^{\alpha_k\beta_k}[l] \ot I_{d_a^R[l]}\,\in\pi_l(\cA)\\
    \Omega^{(l)}&:=\bigoplus_{a=1}^{r} I_{d_a} \ot \Omega_{a}^{(l)}\,\in\pi_l(\cA)'\,.
\end{align*}
Note that the above definitions of $ \pi_l\bk{{\hat w}_k^{\alpha_k\beta_k}[l]}$ and $\Omega^{(l)}$ are not unique as one can always consider $z\pi_l\bk{{\hat w}_k^{\alpha_k\beta_k}[l]}z^{-1}\Omega^{(l)}$ for any invertible $z\in \pi_l(\cA)\cap\pi_l(\cA)'$ and redefine them accordingly. 
Expanding ${\hat w}^{\alpha_k\beta_k}_k[l]$ in a basis $\cB_\cA:=\{x\}$ of $\cA$ as ${\hat w}^{\alpha_k\beta_k}_k[l]=\sum_{x\in\cB_\cA}\delta_x\bk{{\hat w}^{\alpha_k\beta_k}_k[l]}x$, where $\{\delta_x\}$ is the dual basis of $\cA^*$, the $l$-concatenated tensor can be written as
\begin{equation*}
    \underbrace{MM...M}_{l}=\sum_{x\in \cB_\cA}\pi^{\ot l}\circ\Delta^{l-1}(x)\Omega^{(l)}\ot\psi_l(\delta_x)N^l\,,
\end{equation*}
with
\begin{align*}
N&:=\bigoplus_kI_{D_k^L}\ot N_k\,, \\
\psi_l(\delta_x)&=\bigoplus_k\psi_l^{(k)}(\delta_x)\ot I_{D_k^R}\,,
\end{align*}
where 
\begin{equation*}
    \psi_l^{(k)}(\delta_x):=\sum_{\alpha_k,\beta_k=1}\delta_x\bk{{\hat w}_k^{\alpha_k\beta_k}[l]}\ket{\alpha_k}\bra{\beta_k}\,.
\end{equation*}
From this, we obtain an expression of the MPDO $\rho^{(L)}$
\begin{equation*}
    \rho^{(L)}=\pi^{\ot L}\circ\Delta^{L-1}\bk{w^{(L)}}\Omega^{(L)}\,,
\end{equation*}
where 
\begin{align*}
w^{(L)}&:=\sum_{x\in\cB_\cA}\tr\bk{\psi_L(\delta_x)N^L}x\\
    &=\sum_k\tr\bk{N_k^L}\sum_{\alpha_k}{\hat w}_k^{\alpha_k\alpha_k}[L]\,\in\cA\,.
\end{align*}

Next, we show the structure of $\Omega^{(L)}$. Recall that $W_k^{\alpha_k\beta_k}[l_1+l_2]=\sum_{\gamma_k}W_k^{\alpha_k\gamma_k}[l_1]\ot W_k^{\gamma_k\beta_k}[l_2]$. Considering the v-CF on both sides implies
\begin{align}
 \bigoplus_{a,b=1}^{r}\sum_{\gamma_k} W_{k,a}^{\alpha_k\gamma_k}[l_1]&\ot W_{k,b}^{\gamma_k\beta_k}[l_2]  \ot \Omega_{a}^{(l_1)}\ot \Omega_{b}^{(l_2)} \nonumber \\
    &\cong\bigoplus_{c=1}^{r} W_{k,c}^{\alpha_k\beta_k}[l_1+l_2] \ot \Omega_{c}^{(l_1+l_2)}\,,\label{eq:W2=WW}
\end{align}
where $\cong$ denotes equivalence up to an isometry, from which we derive the following lemma.
\begin{lem}\label{lem:Vab}For any $l_1,l_2\in\mathbb{N}$, there exist isometries $V_{a,b}$ such that 
\begin{align}
  V_{a,b}&\bk{\sum_{\gamma_k}W_{k,a}^{\alpha_k\gamma_k}[l_1]\ot W_{k,b}^{\gamma_k\beta_k}[l_2]}V_{a,b}^\dagger\nonumber\\
  &\hspace{1cm}= \bigoplus_cW_{k,c}^{\alpha_k\beta_k}[l_1+l_2]\otimes \Gamma_{ab}^c[l_1+l_2]\,,\label{eq:Vab}
\end{align}
where each $\Gamma_{ab}^c[l_1+l_2]$ is a diagonal matrix, taking the value $\Gamma_{ab}^c[l_1+l_2]>0$ or $\Gamma_{ab}^c[l_1+l_2]=0$ depending on the choice of $(a,b,c)$. Moreover, $\Gamma_{ab}^c[l_1+l_2]$ is symmetric in the sense that
\begin{equation}
    \Gamma_{ab}^c[l_1+l_2]\cong\Gamma_{ab}^c[l_2+l_1]\,.
\end{equation}
\end{lem}
The proof of this lemma is a slight modification of that of Theorem 4.14 in Ref.~\cite{CIRAC2017100}. The only distinction is that each BNT element now depends on the length $l_1,l_2$. 

By inserting Eq.~\eqref{eq:Vab} into Eq.~\eqref{eq:W2=WW}, we obtain
\begin{align}\label{eq:Omegadecomposition}
    \Omega_c^{(l_1+l_2)}&\cong\bigoplus_{a,b=1}^r\Gamma_{ab}^c[l_1+l_2]\ot \Omega_{a}^{(l_1)}\ot \Omega_{b}^{(l_2)}\,.
\end{align}
We can iteratively apply Eq.~\eqref{eq:Omegadecomposition} to obtain that
\begin{align*}
    \Omega_c^{(l)}&\cong\bigoplus_{a_1,b_1=1}^r\Gamma_{a_1b_1}^{c}[1+(l-1)]\ot\Omega_{a_1}\otimes\Omega_{b_1}^{(l-1)}\\
    &\cong\bigoplus_{{\bf a},{\bf b}}\Gamma_{a_1b_1}^{c}[1+(l-1)]\ot\Gamma_{a_2b_2}^{b_1}[1+(l-2)]\\
    &\quad\ot\Omega_{a_1}\ot\Omega_{a_2}\ot\Omega_{b_2}^{(l-2)}\\
    &\cong\cdots\\
    &\cong\Gamma^c_l\Omega^{\ot l}\,,
\end{align*}
where
\begin{align*}
    \Gamma^c_l&:=\bigoplus_{\bf a,b}\Gamma_{a_1b_1}^{c}[1+(l-1)]\ot\Gamma_{a_2b_2}^{b_1}[1+(l-2)]\ot\\
    &\qquad\cdots\ot \Gamma_{a_{l-1}a_l}^{b_{l-2}}[1+1]\bk{\bigotimes_{k=1}^L I_{d_{a_k}^R}}\,,\\
    \Omega&:=\bigoplus_a\Omega_a \,,
\end{align*}
which completes the proof by renaming as $\Gamma_{k,ab}^{c}:=\Gamma_{a_{l-1}a_l}^{b_{l-2}}[1+(l-k)]$. 

\end{document}